\documentclass{amsart}
\usepackage[T1]{fontenc}
\usepackage{amsmath,amssymb,amsthm,anyfontsize,bm,dcolumn,diagbox,enumerate,enumitem,etoolbox,float,graphicx,hhline,mathptmx,mathrsfs,mathtools,amsrefs,newtxmath,newtxtext,upref,xcolor,xurl, braket}
\usepackage[a4paper, margin = 1in]{geometry}
\usepackage{hyperref}
\hypersetup{
colorlinks=true,
linkcolor=blue,
citecolor=red}
\usepackage[capitalise]{cleveref}
\usepackage{tikz}
\usepackage{tikz-cd}

\newtheorem{theorem}{Theorem}[section]

\newtheorem{cor}[theorem]{Corollary}
\newtheorem{lemma}[theorem]{Lemma}
\newtheorem{prop}[theorem]{Proposition}

\theoremstyle{definition} 
\newtheorem{definition}[theorem]{Definition}

\newtheorem{quest}[theorem]{Question}

\theoremstyle{remark}

\numberwithin{equation}{section}

\DeclareMathOperator{\tr}{tr}
\DeclareMathOperator{\Tr}{Tr}
\DeclareMathOperator{\vect}{vec}
\DeclareMathOperator{\mat}{mat}
\DeclareMathOperator{\aux}{aux}
\DeclareMathOperator{\fin}{fin}
\DeclareMathOperator{\op}{op}

\newcommand{\norm}[1]{\left\lVert #1 \right\rVert}
\newcommand{\abs}[1]{\left\lvert #1\right\rvert}

\title{Robust Self-Testing for Synchronous Correlations and Games}
\author{Prem Nigam Kar}
\address{DTU Compute, Technical University of Denmark, 2800 Kongens Lyngby, Denmark.}
\email{pkar@dtu.dk}
\thanks{All authors are supported by Carlsberg Semper Ardens Accelerate CF21-0682 Quantum Graph Theory.}
\date{}

\begin{document}

\begin{abstract}
We develop an abstract operator-algebraic characterization of robust self-testing for synchronous correlations and games. Specifically, we show that a synchronous correlation is a robust self-test if and only if there is a unique state on an appropriate $C^*$-algebra that ``implements" the correlation. Extending this result, we prove that a synchronous game is a robust self-test if and only if its associated $C^*$-algebra admits a unique amenable tracial state. This framework allows us to establish that all synchronous correlations and games that serve as commuting operator self-tests for finite-dimensional strategies are also robust self-tests. As an application, we recover sufficient conditions for linear constraint system games to exhibit robust self-testing. We also demonstrate the existence of a synchronous nonlocal game that is a robust self-test but not a commuting operator self-test, showing that these notions are not equivalent. 
\end{abstract}

\maketitle

\section{Introduction}

Self-testing, first introduced by Mayers and Yao~\cite{mayers2004self}, involves certification of quantum devices from "classical" data such as measurement statistics (see~\cite{vsupic2020self} for a review). Over the years, it has developed into an active field of research with a range of applications including device-independent cryptography~\cites{mayers1998quantum, mayers2004self}, delegated quantum computation~\cite{coladangelo2024verifier}, entanglement detection~\cites{bowles2018device, bowles2018self}, and quantum complexity theory~\cites{fitzsimons2019quantum, natarajan2017quantum, natarajan2019neexp}. It also played a major role in the recent breakthrough result $\textsf{MIP}^*= \textsf{RE}$~\cite{ji2021mip}, which lead to the resolution of longstanding open problems of Tsirelson, Connes, and Kirchberg.  

A game or a correlation is said to be a self-test if there is a unique strategy that is optimal for the game or if there is a unique strategy producing the correlation, up to local isometries. However, in order for self-tests to be useful for practical applications, they need to be robust. In other words, we require that any strategy that is near optimal for the game, or a strategy that produces a correlation close the prescribed one, is close to the ideal or honest strategy up to local isometries. 

While it is often much more difficult to show that a correlation or a game is a robust self-test than it is to show that it is a self-test, most attempts to do so have been successful. Until recently, all known examples of self-tests were also known to be robust. This prompted Man\v{c}inska and Schmidt \cite{mancinska_counterexamples_2023} to investigate if all self-tests are robust, a question which they answered in the negative by constructing a game that is a self-test, but is not a robust self-test, thus showing that robust self-testing is a strictly rarer phenomenon. However, this counterexample is carefully constructed with the intention of being a self-test that is not robust. Hence, one can ask the following natural question: 

\begin{quest}\label{quest-1}
    Is there a class of nonlocal games or correlations, for which all self-tests are also robust? Is there an appropriate natural definition of self-testing under which all self-tests are also robust?
\end{quest} 

One widely successful approach to showing that a correlation or a game is a robust self-test has been to associate optimal strategies with representations of an appropriate algebraic object. In particular, a $C^*$-algebraic characterization of self-testing was obtained in \cite{paddock_operator-algebraic_2023}, where it was shown that a correlation (under certain reasonable assumptions) is a self-test if and only if there is a unique finite-dimensional state on the $C^*$-algebra $\mathcal{A}_{XA}^{POVM} \otimes_{\min} \mathcal{A}_{XA}^{POVM}$ associated with the correlation. This approach is well suited to tackle \cref{quest-1} as it can reduce this question to a problem about states on $C^*$-algebras.

Hence, with the intention of making progress towards \cref{quest-1}, it is natural to ask if robust self-testing can be characterized in terms of "abstract state self-testing". Formally, we ask:  

\begin{quest}\label{quest-2}
    Is it possible to characterize robust self-testing in terms of abstract state self-testing over some appropriate subset of states on $\mathcal{A}_{XA}^{POVM} \otimes_{\min}\mathcal{A}_{XA}^{POVM}$?     
\end{quest}

Synchronous nonlocal games and Linear Constraint System games are two prominent and well studied classes of nonlocal games. Moreover, their perfect strategies can be characterized in terms of existence of tracial states on the $C^*$-algebra of a synchronous game, or in terms of existence of appropriate representations of the solution group in the case of LCS games, making them prime candidates for \cref{quest-1}. It is (relatively) easy to show that if the algebraic objects associated to these games admit a unique state or representation, then these games are self-tests. As a more concrete approach to \cref{quest-1}, we ask the following:

\begin{quest}\label{quest-3}
    If the $*$-algebra of a synchronous game admits a unique tracial state, or if the solution group of an LCS game admits a unique representation, then are these games robust self-tests?
\end{quest}

A salient feature of the $C^*$-algebraic approach to self-testing introduced in \cite{paddock_operator-algebraic_2023} is that it easily generalises to a definition of self-testing for the commuting operator model. All known examples of self-tests are also known to be commuting operator self-tests. Hence, in light of the separation of tensor-product quantum correlations and quantum commuting correlations obtained by \cite{ji2021mip}, one natural question is to ask the following: 

\begin{quest}\label{quest-4}
    Are all robust self-tests also commuting operator self-tests?
\end{quest}

We set out to answer these questions, but first we state our results more formally. 

\subsection*{Results}

Our main result, answering \cref{quest-2} for synchronous correlations, is an extension of abstract state self-testing to the robust case formalized in the following theorem: 
\begin{theorem}[Main Result]\label{thm:main-result}
    Let $\widetilde{p} = \{\widetilde{p}(a,b\mid x,y)\}_{x,y \in X, a, b \in A}$ be a synchronous quantum correlation that is an extreme point of $C_q(X,X,A,A)$. Then, the following are equivalent: 
    \begin{enumerate}[label = \roman*.]
        \item  There is a unique state $\widetilde{\rho}$ on $\mathcal{A}_{XA}^{POVM} \otimes_{\min} \mathcal{A}_{XA}^{POVM}$ that implements $\widetilde{p}$. 
        \item The correlation $\widetilde{p}$ robustly self-tests a strategy $(\mathbb{C}^d, \mathbb{C}^d, \{\widetilde{E}_{x,a}\}, \{\widetilde{F}_{y,b}\}, \ket{\widetilde{\psi}}).$ 
    \end{enumerate}
\end{theorem}
This result generalises the ones obtained in \cites{mancinska_constant-sized_2024, Volcic2024constantsizedself}. It should be noted that since we assume that the correlation $p$ is in $C_q(X,Y,A,B)$, the unique state $\widetilde{\rho}$ will always be finite-dimensional. This also implies that if a synchronous correlation $p$ that is an extreme point of $C_q(X,X,A,A)$ is a commuting operator self-test is also a robust self-test, thus answering \cref{quest-1} for the synchronous case. We remark here that the implication \textit{ii.} $\implies$ \textit{i.} always holds true, i.e.,  we have: 

\begin{prop}\label{prop:robust-self-test-implies-abstract-self-test}
    Let $\widetilde{p} \in C_q(X,Y,A,B)$ be a correlation that robustly self-tests a finite dimensional strategy $\widetilde{S}$. If $\widetilde{\rho}$ is the abstract state on $\mathcal{A}_{XA}^{POVM} \otimes_{\min} \mathcal{A}_{XA}^{POVM}$ given by $\widetilde{S}$, then $\widetilde{\rho}$ is the unique state on $\mathcal{A}_{XA}^{POVM} \otimes_{\min} \mathcal{A}_{XA}^{POVM}$ implementing $\widetilde{p}$.  
\end{prop}

\cref{thm:main-result} enables us to obtain several interesting corollaries the first of which is obtaining an abstract characterization of self-testing for synchronous games in \cref{thm:synch-game-self-test-cond}. We also show the existence of a synchronous nonlocal game that is a robust self-test for finite dimensional strategies, but that is not a commuting operator self-test, thus answering \cref{quest-4} in the negative. We rely heavily on the operator algebraic characterization of self-testing to arrive at this result. Formally, we have the following: 

\begin{theorem}\label{thm:counter-example}
There exists a nonlocal game $\mathscr{G}$ that is a robust self-test, but is not a commuting operator self-test. 
\end{theorem}

\subsection*{Acknowledgements}
It has come to our attention that similar results have been obtained independently in the recent work~\cites{zhao2024robust, yuming-thesis}. We would like to thank to Laura Man\v{c}inska and David Roberson for several fruitful discussions, without which this project would not have been possible. We are also thankful to Yuming Zhao for pointing out that the $C^*$-algebra $\mathcal{A}_{XA}^{POVM} \otimes \mathcal{A}_{XA}^{POVM}$ is residually finite dimensional. 

\section{Preliminaries}

By a \emph{correlation} $p \in C(X, Y, A, B)$, we mean a joint conditional probability distribution of the form $(p(a,b \mid x,y))_{x \in X, y \in Y,a \in A, b \in B}$. A correlation $p \in C(X, Y, A, B)$ is said to be \emph{synchronous} if $X = Y$, $A = B$ and $p(a,b \mid x,x) = 0$ whenever $a \neq b$. A correlation $p \in C(X,Y,A,B)$ is said to be a \emph{quantum spatial (qs) correlation} if there are POVMs $\{E_{x,a}\}_{a\in A} \subseteq \mathcal{B(H}_A)$, $\{F_{y,b}\}_{b \in B} \subseteq \mathcal{B(H}_B)$ for each $x \in X$ and $y \in Y$ and a state $\ket{\psi} \in \mathcal{H}_A \otimes \mathcal{H}_B$, for some Hilbert spaces $\mathcal{H}_A$ and $\mathcal{H}_B$ such that 
$p(a,b \mid x,y) = \braket{\psi \mid E_{x,a} \otimes F_{y,b}\mid \psi} \text{ for all } x \in X, y \in Y, a \in A, b \in B$

We shall call $S= (\mathcal{H}_A, \mathcal{H}_B, \{E_{x,a}\}_{x \in X, a \in A}, \{F_{y,b}\}_{y \in Y, b \in B}, \ket{\psi})$ a \emph{tensor-product strategy} for $p$. Furthermore, if the Hilbert spaces $\mathcal{H}_A$ and $\mathcal{H}_B$ are finite-dimensional, we shall call $S$ a \emph{finite-dimensional strategy} for $p$, and a qs correlation with a finite-dimensional strategy is called a \emph{quantum correlation}. We shall denote the set of all qs and quantum correlations in $C(X,Y,A,B)$ by $C_{qs}(X,Y,A,B)$ and $C_{q}(X,Y,A,B)$ respectively. 

A \emph{qc strategy} for a correlation $p \in C(X,Y,A,B)$ is a $4$-tuple $(\mathcal{H}, \{E_{x,a}\}_{x\in X, a \in A}, \{F_{y,b}\}_{y\in Y, b \in B}, \ket{\psi})$, where $\mathcal{H}$ is a Hilbert space $\{E_{x,a}\}_{a\in A}\subseteq \mathcal{B(H})$, $\{F_{y,b}\}_{b \in B}\subseteq\mathcal{B(H})$ are mutually commuting POVMs and $\ket{\psi}$ is a state such that $p(a,b \mid x,y) = \braket{\psi \mid E_{x,a} F_{y,b}\mid \psi}$ for all $x \in X$, $y \in Y$, $a \in A$ and $b \in B$. A correlation is said to be a \emph{quantum commuting correlation} if it has a qc strategy. Note that any tensor product strategy $S= (\mathcal{H}_A, \mathcal{H}_B, \{E_{x,a}\}, \{F_{y,b}\}, \ket{\psi})$ can be represented as the qc strategy $S' = (\mathcal{H}_A \otimes \mathcal{H}_B, \{E_{x,a} \otimes I_{\mathcal{H}_B}\}, \{I_{\mathcal{H}_A} \otimes F_{y,b}\}, \ket{\psi} )$. Hence, one has $C_{qs}(X,Y,A,B) \subseteq C_{qc}(X,Y,A,B)$. 

A finite-dimensional strategy $\widetilde{S} = \{\mathcal{H}, \mathcal{H}, \{\widetilde{E}_{x,a}\}, \{\widetilde{F}_{y,b}\}, \ket{\psi}\}$ for a correlation $\widetilde{p} \in C_q(X,X, A, A)$ such that: 
\begin{enumerate}[label = \roman*.]\label{cond-PME}
        \item $\widetilde{E}_{x,a}$ is a projection for all $x \in X$, $a \in A$,
        \item $\widetilde{F}_{y,b} = \widetilde{E}_{y,b}^T$ for all $y \in X$, $b \in B$, and 
        \item $\ket{\psi} \in \mathcal{H} \otimes \mathcal{H}$ is the maximally entangled state.
\end{enumerate}  
    is known as a \emph{Projective Maximally Entangled (PME)} strategy. Let $\widetilde{p}$ be the correlation induced by the PME strategy $\widetilde{S}$. Then, one has 
\begin{equation}
\widetilde{p}(a,b\mid x,y) = \braket{\psi \mid \widetilde{E}_{x,a} \otimes \widetilde{F}_{y,b} \mid \psi} = \tr(\widetilde{E}_{x,a}\widetilde{E}_{y,b}) = \tr(\widetilde{F}_{x,a}\widetilde{F}_{y,b}).
\end{equation}

In particular, it is not too difficult to see that $\widetilde{p}$ is a synchronous correlation. We shall denote the set of all synchronous correlation in $C_t(X,X,A,A)$ by $C_t^s(X,A)$ for all $t \in \{q,qs,qa,qc\}$.
\subsection{Nonlocal Games}

A nonocal game is a cooperative game involving two players: Alice and Bob,  and a verifier. Formally, a \emph{nonlocal game} $\mathscr{G} = (X, Y, A, B, V)$ is a $5$-tuple, where $X, Y, A, B$ are finite sets and $V: X \times Y \times A \times B \to \{0,1\}$ is a function called the \emph{predicate}. A nonlocal game is played as follows: 
\begin{itemize}
    \item In each round, the verifier picks two questions (according to a uniform probability distribution) $x\in X$ and $y \in Y$,  and sends them to Alice and Bob respectively. 
    \item The players choose their responses $a \in A$ and $ b \in B$ and send them to the verifier. 
    \item The players win this round of the game if $V(x,y,a,b) = 1$, and they lose otherwise.
\end{itemize}

The players are not allowed to communicate after they receive their questions. However, they may agree on a strategy to win the game. It is customary in the literature to write $V(a,b \mid x,y)$ for $V(x,y,a,b)$, and we shall adopt this custom for the rest of the work. We now describe various types of strategies the players can use to try and win a nonlocal game.  

A \emph{deterministic} strategy is an ordered pair of functions $(f_A, f_B)$ where $f_A: X \to A$ and $f_B: Y \to B$ encode Alice's and Bob's strategies respectively. In other words, upon receiving inputs $x \in X$ and $y \in Y$, the players respond with $f_A(x)$ and $f_B(y)$ respectively. Such a strategy is said to be perfect if $V(f_A(x), f_B(y) \mid x, y) = 1$ for all $ x\in X$ and $y \in Y$. 

The players can also make use of \emph{probabilistic strategies}, i.e. correlations $p \in C(X, Y, A, B)$. Such a correlation $p \in C(X,Y,A,B)$ is said to be \emph{perfect} or \emph{winning} if $p(a,b \mid x,y) = 0$, whenever $V(a,b \mid x,y) = 0$. Probabilistic strategies making use of shared local randomness are known as \emph{local strategies}. Since local strategies can be shown to be convex combinations of deterministic strategies, they offer no competitive advantage over the set of deterministic strategies. 

Probabilistic strategies that employ quantum, qs, and qc correlations are known as quantum, qs, and qc strategies respectively. It is well known that there exist nonlocal games with a quantum advantage, i.e., there are nonlocal games which have no perfect local strategies, but have perfect quantum strategies. If we let $C_{qa}(X,Y,A,B)$ denote the closure of $C_q(X,Y,A,B)$ and $C_{loc}(X,Y,A,B)$, then one has the following hierarchy of classes of correlations: 
$$C_{loc}(X,Y,A,B) \subseteq C_{q}(X,Y,A,B) \subseteq C_{qs}(X,Y,A,B) \subseteq C_{qa}(X, Y, A, B) \subseteq C_{qc}(X, Y, A, B).$$
It is also known from a series of results~\cites{coladangelo2020inherently, slofstra2019set, slofstra2020tsirelson, ji2021mip} that all of these inclusions are strict. We also remark here that each of the above sets in convex, with $C_{loc}^s(X,Y,A,B), C_{qa}(X,Y,A,B)$ and $C_{qc}(X,Y,A,B)$ also being closed. The same holds true for the synchronous versions of these correlation classes, i.e.
$$C_{loc}^s(X,A) \subseteq C_{q}^s(X,A) \subseteq C_{qs}^s(X,A) \subseteq C_{qa}^s(X, A) \subseteq C_{qc}^s(X,A).$$, 
with all the inclusions being strict. Moreover, each of these sets of synchronous correlations is convex, with $C_{loc}^s(X,A), C_{qa}^s(X,A),$ and $C_{qc}^s(X,A)$ also being closed. 

We shall be mostly interested in synchronous nonlocal games and Linear Constraint System (LCS) games. We shall provide the necessary background on synchronous games in the next subsection. We refer the readers to \cites{coladangelo2019robustselftestinglinearconstraint, slofstra2019set, slofstra2020tsirelson} for the necessary prerequisites on LCS games.

A nonlocal game $\mathscr{G} = (X, Y, A, B, V)$ is said to be synchronous if $X = Y$, $A = B$, and $V(a, b \mid x,x) = 0$ whenever $ a\neq b$. we shall denote a synchronous nonlocal game $\mathscr{G}$ by the 3-tuple $(X,Y,V)$, where $V : X \times X \times Y \times Y \to \{0,1\}$ is understood to be the predicate. It is evident that a perfect correlation for a synchronous game is always a synchronous correlation. 

\subsection{Self-testing}

Let $\widetilde{S} = (\widetilde{\mathcal{H}}_A, \widetilde{\mathcal{H}}_B, \{\widetilde{E}_{x,a}\}, \{\widetilde{F}_{y,b}\}, \ket{\widetilde{\psi}})$ and $S = (\mathcal{H}_A, \mathcal{H}_B, \{E_{x,a}\}, \{F_{y,b}, \ket{\psi}\})$ be two finite-dimensional strategies. We say that $\widetilde{S}$ is a \emph{local dilation} of $S$ if there are isometries $V_A: \mathcal{H}_A \to \widetilde{\mathcal{H}}_A \otimes \mathcal{K}_A$ and $V_B: \mathcal{H}_B \to \widetilde{\mathcal{H}}_B \otimes \mathcal{K}_B$ for some Hilbert spaces $\mathcal{K}_A$ and $\mathcal{K}_B$, and a state $\ket{\psi_{\aux}} \in \mathcal{K}_A \otimes \mathcal{K}_B$ such that 
\begin{align}\label{eq:self-test-measurement-and-state}
    V_A \otimes V_B (E_{x,a} \otimes F_{y,b}) \ket{\psi} = ((\widetilde{E}_{x,a} \otimes I_{\mathcal{K}_A}) \otimes (F_{y,b} \otimes I_{\mathcal{K}_B}))\ket{\widetilde{\psi}}\ket{\psi_{\aux}}
\end{align}
for all $x \in X$, $y \in Y$, $a \in A$ and $b \in B$. Summing up over $a,b$ in \cref{eq:self-test-measurement-and-state}, we get: 
\begin{align}\label{eq:self-test-state}
    V_A \otimes V_B \ket{\psi} = \ket{\widetilde{\psi}}\ket{\psi_{\aux}}
\end{align}

Note that if $\widetilde{S}$ is a local dilation of $S$, then they induce the same correlation. We shall say that a correlation $p \in C_q(X, Y, A, B)$ \emph{self-tests} a strategy $\widetilde{S}$ if $\widetilde{S}$ is a local dilation of every strategy $S$ inducing the correlation $p$. We now introduce the robust versions of these defintions: 

A finite-dimensional quantum strategy $\widetilde{S} = (\widetilde{\mathcal{H}}_A, \widetilde{\mathcal{H}}_B, \{\widetilde{E}_{x,a}\}, \{\widetilde{F}_{y,b}\}, \ket{\widetilde{\psi}})$ is said to be a \emph{local $\epsilon$-dilation} of another finite-dimensional quantum strategy $S = (\mathcal{H}_A, \mathcal{H}_B, \{E_{x,a}\}, \{F_{y,b}, \ket{\psi}\})$ if there are isometries $V_A: \mathcal{H}_A \to \widetilde{\mathcal{H}}_A \otimes \mathcal{K}_A$ and $V_B: \mathcal{H}_B \to \widetilde{\mathcal{H}}_B \otimes \mathcal{K}_B$ for some Hilbert spaces $\mathcal{K}_A$ and $\mathcal{K}_B$, and a state $\ket{\psi_{\aux}} \in \mathcal{K}_A \otimes \mathcal{K}_B$ such that 
\begin{align}\label{eq:robust-self-test-measurement-and-state}
    V_A \otimes V_B (E_{x,a} \otimes F_{y,b}) \ket{\psi} \approx_{\epsilon} ((\widetilde{E}_{x,a} \otimes I_{\mathcal{K}_A}) \otimes (F_{y,b} \otimes I_{\mathcal{K}_B}))\ket{\widetilde{\psi}}\ket{\psi_{\aux}}
\end{align}
for all $x \in X$, $y \in Y$, $a \in A$ and $b \in B$ and 
\begin{align}\label{eq:robust-self-test-state}
    V_A \otimes V_B \ket{\psi} \approx_{\epsilon} \ket{\widetilde{\psi}}\ket{\psi_{\aux}}
\end{align}
 A correlation $\widetilde{p} \in C_q(X,Y,A,B)$ is said to be a \emph{robust self-test} for a strategy $\widetilde{S}$ if for every $\epsilon > 0$, there is a $\delta>0$ such that if $\norm{p - \widetilde{p}}_1$ for some correlation $p \in C_q(X, Y, A, B)$ with strategy $S$, then $\widetilde{S}$ is a local $\epsilon$-dilation of $S$. These definitions naturally extend to nonlocal games. 

\begin{definition}\label{def:nonlocal-game-self-test}
     A nonlocal game $\mathscr{G} = (X, Y, A, B, V)$ is said to be a \emph{self-test} for a strategy $\widetilde{S}$ if for every correlation $p \in C_q(X, Y, A, B)$ with a model $S$ that is a perfect strategy for $\mathscr{G}$, $\widetilde{S}$ is a local dilation of $S$. 
     
     Similarly, $\mathscr{G}$ is said to be a \emph{robust self-test} for $\widetilde{S}$ if $\mathscr{G}$ is a self-test for $\widetilde{S}$ and for every $\epsilon > 0$ there is a $\delta > 0$ such that for  every correlation $p \in C_q(X, Y, A, B)$ with a model $S$ such that 
     $$\sum_{\substack{x \in X, y \in Y, a \in A, b \in B  \\
     \text{such that } V(a,b \mid x, y = 0)}} p(a,b \mid x, y) \leq \delta, $$ 
     one has that $\widetilde{S}$ is a local-$\epsilon$ dilation of $\widetilde{S}$. 
\end{definition}
 
\subsection{Background on Operator Algebras}

We shall assume basic familiarity with the theory of operator algebras. We refer the reader to \cites{kadison1997elementary, brown88c} for a detailed exposition. Recall that a \emph{state} on a $C^*$-algebra is a positive linear functional with norm $1$. A state $\rho$ on a $C^*$-algebra $\mathcal{A}$ is said to be a \emph{tracial state} or just a \emph{trace} if $\rho(ab) = \rho(ba)$ for all $a, b \in \mathcal{A}$. Given a state $\rho$ on a $C^*$-algebra the \emph{GNS}-construction gives a Hilbert space $\mathcal{H}_{\rho}$, a $\ast$-representation $\pi_{\rho}: \mathcal{A} \to \mathcal{B(H}_{\rho})$ and a vector $\ket{\eta_{\rho}} \in \mathcal{H}$ such that $\rho(a) = \braket{ \eta_{\rho} \mid \pi_\rho(a)\mid \eta_{\rho}}$ for all $a \in \mathcal{A}$. Hence, the GNS construction represents an abstract state on a $C^*$-algebra as a concrete vector state. 

The vector $\ket{\eta_{\rho}}$ is also known to be \emph{cyclic} for the representation $\pi_\rho$, i.e. $\pi_{\rho}(\mathcal{A})\ket{\eta_{\rho}}$ is dense in $\mathcal{H}_{\rho}$. The GNS construction is unique in the sense that for any other cyclic representation $(\pi, \ket{\eta})$ for some $\ket{\eta} \in \mathcal{H}$ such that $\rho(a) =\braket{\eta\mid \pi(a) \mid \eta}$ for all $a \in \mathcal{A}$, there is a unitary $U : \mathcal{H} \to \mathcal{H}_{\rho}$ such that $U \ket{\eta} = \ket{\eta_{\rho}}$ and $U \pi(a) U^* = \pi_\rho(a)$ for all $ a\in \mathcal{A}$. 

The \emph{$C^*$-algebra of a synchronous game}~\cites{helton2017algebras, lupini_perfect_2020} is a structure that encodes all information about a synchronous game. Formally, given a synchronous nonlocal game $\mathscr{G} = (X, Y, V)$, the $C^*$-algebra of $\mathscr{G}$, denoted by $C^*(\mathscr{G})$, is the universal $C^*$-algebra generated by projections $\{e_{x,a}\}_{x\in X, a \in A}$ such that: 
\begin{enumerate}[label = \roman*.]
    \item $\sum_a e_{x,a} = 1$ for all $x \in X$ and
    \item $e_{x,a}e_{y,b} = 0$ whenever $V(a,b \mid x, y) = 0$
\end{enumerate}
We note that this $C^*$-algebra is not always non-trivial. 

It is known that existence of strategiies for $\mathscr{G}$ can be formulated in terms of existence of strategies on $C^*(\mathscr{G})$. In particular, we have the following result: 

\begin{prop}[\cite{helton2017algebras}]
    Let $\mathscr{G}$ be a synchronous game. Then, the following hold true: 
    \begin{enumerate}[label = \roman*.]
        \item $\mathscr{G}$ has a quantum strategy if and only if $C^*(\mathscr{G})$ admits a finite-dimensional tracial state.
        \item $\mathscr{G}$ has a qa strategy if and only if $C^*(\mathscr{G})$ admits an amenable tracial state.
        \item $\mathscr{G}$ has a quantum commuting strategy if and only if $C^*(\mathscr{G})$ admits a tracial state.
    \end{enumerate}
\end{prop}

\subsection{Abstract State Self-Testing}

Let $X$ and $A$ be two finite sets, and let $\mathcal{A}_{XA}^{POVM}$ be the universal unital $C^*$-algebra generated by positive contractions $\{e_{x,a}\}_{x\in X, a \in A}$ such that $\sum_{a\in A} e_{x,a} = 1$ for all $x \in X$. Then, any tensor product quantum strategy $S = (\mathcal{H}_A, \mathcal{H}_B, \{E_{x,a}\}, \{F_{y,b}\}, \ket{\psi})$ induces representations $\pi_A: \mathcal{A}_{XA}^{POVM} \to \mathcal{B(H}_A)$ and $\pi_B: \mathcal{A}_{XA}^{POVM} \to \mathcal{B(H}_B)$. 

Hence, each tensor product strategy $S$ for a correlation $p \in C_q(X,Y,A,B)$ can be associated with the state $\rho_S: \mathcal{A}_{XA}^{POVM} \otimes_{\min} \mathcal{A}_{XA}^{POVM} \to \mathbb{C}$, defined by $\rho_S(a) = \braket{\psi \mid \pi_A \otimes \pi_B (a)\mid \psi}$. Such a state is called \emph{finite-dimensional} if the associated strategy is finite-dimensional. We say that a state $\rho$ on $\mathcal{A}_{XA}^{POVM} \otimes_{\max} \mathcal{A}_{YB}^{POVM}$ \emph{implements} $p$ if $\widetilde{\rho}(e_{x,a}\otimes f_{y,b}) = \widetilde{p}(a,b\mid x,y)$ for all $x,y \in X$ and $a,b \in A$. The authors of \cite{paddock_operator-algebraic_2023} show that a synchronous correlation $\widetilde{p}$ that is an extreme point of $C_q(X,X,A,A)$ is a self-test for finite dimensional strategies if and only if it is an \emph{abstract state self-test} for finite dimensional states, i.e. if 

\begin{equation}\label{self-test-condition}
    \text{There is a unique finite-dimensional state $\rho$ on $ S(\mathcal{A}_{XA}^{POVM} \otimes_{\min} \mathcal{A}_{XA}^{POVM})$  that implements $\widetilde{p}$.}
\end{equation}

From this point, given a convex set $S$, we shall denote the set of its extreme points as $\partial S$. We shall show that a synchronous correlation $p \in \partial C_q^S(X,A)$, is a robust self-test if and only if it an abstract state self-test for all states on $\mathcal{A}_{XA}^{POVM} \otimes_{\min} \mathcal{A}_{XA}^{POVM}$, i.e. if 

\begin{equation}\label{robust-self-test-condition}
    \text{There is a unique state on $ S(\mathcal{A}_{XA}^{POVM} \otimes_{\min} \mathcal{A}_{XA}^{POVM})$ that implements $\widetilde{p}$.}
\end{equation}
 
One can also show that each qc strategy for a correlation $p \in C(X, Y, A, B)$ gives a state on $\mathcal{A}_{XA}^{POVM} \otimes_{\max} \mathcal{A}_{XA}^{POVM}$. As a generalisation of the abstract state self-testing definition from the finite-dimensional state case, we say that a correlation $p$ is an \emph{commuting operator self-test} (as defined in \cite{paddock_operator-algebraic_2023}) if there is a unique state $\rho$ on $\mathcal{A}_{XA}^{POVM} \otimes_{\max} \mathcal{A}_{XA}^{POVM}$ such that $\rho(e_{x,a} \otimes e_{y,b}) = p(a,b \mid x,y)$ for all $x \in X$, $y \in Y$, $a \in A$ and $b \in B$. Since each state on $\mathcal{A}_{XA}^{POVM} \otimes_{\min} \mathcal{A}_{XA}^{POVM}$ extends to a state on $\mathcal{A}_{XA}^{POVM} \otimes_{\max} \mathcal{A}_{XA}^{POVM}$, we see that being a commuting operator self is a stronger requirement that being an abstract state self-test for (finite-dimensional) states on $\mathcal{A}_{XA}^{POVM} \otimes_{\min} \mathcal{A}_{XA}^{POVM}$.

\section{Characterising Synchronous Self-tests} 

Next, we fix the "ideal" or "honest" strategy $\widetilde{S}$ that will be self-tested by a correlation $\widetilde{p} \in \partial C_q^s(X,A)$ satisfying the hypothesis of \cref{thm:main-result} and condition~\ref{robust-self-test-condition}:

\begin{lemma}[\cite{paddock_operator-algebraic_2023}]\label{lem:canonical-PME-strat}
    Let $\widetilde{p} \in \partial C_q^s(X,X,A,A) \cap C_q(X,A)$ be a correlation satisfying condition~\ref{self-test-condition}. Then, there is a projective maximally entangled (PME) strategy $(\mathbb{C}^d, \mathbb{C}^d, \{\widetilde{E}_{x,a}\}, \{\widetilde{F}_{y,b}\}, \ket{\widetilde{\psi}})$ for $\widetilde{p}$ such that the associated representation $\widetilde{\pi}_A \otimes \widetilde{\pi}_B$ on $\mathcal{A}_{XA}^{POVM} \otimes \mathcal{A}_{XA}^{POVM}$ is irreducible. 
\end{lemma}

In particular, $(\widetilde{\pi}_A \otimes \widetilde{\pi}_B, \ket{\widetilde{\psi}})$ is a GNS representation of $\widetilde{\rho}$. We need a result about GNS representations of states that we prove.

\begin{lemma}\label{lem:kernel-GNS-rep}
    Let $\mathcal{A}$ be a unital $C^*$-algebra and $\rho$ be a tracial state on $\mathcal{A}$. Let $\mathcal{J}_{\rho} \coloneq \{a \in \mathcal{A} \mid \rho(a^*a) = 0\}$ and $\pi_{\rho}: \mathcal{A} \to \mathcal{B(H}_\rho)$ be the GNS representation of the state $\rho$. Then, $\mathcal{J}_{\rho}$ is the kernel of $\pi_{\rho}$. 
\end{lemma}
\begin{proof}
    Since $\rho$ is a tracial state, $\mathcal{J}_{\rho}$ is a closed two-sided ideal of $\mathcal{A}$. Recall that $\mathcal{H}_{\rho}$ is constructed by completing $\mathcal{A}/\mathcal{J}_{\rho}$ under the norm induced by the following inner-product: $$\braket{a+ \mathcal{J}_\rho, b + \mathcal{J}_{\rho}} = \rho(a^*b).$$ The image of $a$ under the representation $\pi_{\rho}$ is given by (the extension to $\mathcal{H}_{\rho}$ of) the linear map $b+\mathcal{J}_\rho \mapsto ab + \mathcal{J}_{\rho}$. In particular, $a \in \ker{\pi_{\rho}}$ if and only if $ab \in \mathcal{J}_{\rho}$ for all $b \in \mathcal{A}$, which finishes the proof. 
\end{proof}

For the rest of this section, we fix $\widetilde{p}$ to be a synchronous quantum correlation that is an extreme point of $C_q(X, X, A, A)$ satisfying condition~\ref{self-test-condition}. We also fix the aforementioned PME strategy for $\widetilde{p}$ to be $\widetilde{S} = (\mathbb{C}^d, \mathbb{C}^d, \{\widetilde{E}_{x,a}\}, \{\widetilde{F}_{y,b}\}, \ket{\widetilde{\psi}})$. The representation of $\mathcal{A}_{XA}^{POVM}$ induced by this strategy will be denoted by $\widetilde{\pi}_A \otimes \widetilde{\pi}_B$. Since $\widetilde{\pi}_A \otimes \widetilde{\pi}_B$ is irreducible, so is each of $\widetilde{\pi}_A$ and $\widetilde{\pi}_B$. 

We now prove a series of lemmas that will help us show that abstract state self-testing and self-testing are equivalent for synchronous quantum correlations $\widetilde{p} \in C_q^s(X, A)$ that are extreme points of $C_q(X, X, A, A)$.

\begin{lemma}\label{lem:UCP-map-from-canon-strat}
     Let $\widetilde{\rho}_A: \mathcal{A}_{XA}^{POVM} \to \mathbb{C}$ and $\widetilde{\rho}_B: \mathcal{A}_{XA}^{POVM} \to \mathbb{C}$ be the states defined by $\widetilde{\rho}_A(a) = \widetilde{\rho}(a \otimes 1)$ and $\widetilde{\rho}_B(b) = \widetilde{\rho}(1 \otimes b)$ respectively. Define $\mathcal{J}_A = \{a \in \mathcal{A}_{XA}^{POVM} \mid \widetilde{\rho}_A(a^*a) = 0\}$ and $\mathcal{J}_B = \{b \in \mathcal{A}_{XA}^{POVM} \mid \widetilde{\rho}_B(b^*b) = 0\}$ and let $q_A: \mathcal{A}_{XA}^{POVM} \to \mathcal{A}_{XA}^{POVM}/ \mathcal{J}_A$ and $q_B: \mathcal{A}_{XA}^{POVM} \to \mathcal{A}_{XA}^{POVM}/ \mathcal{J}_B$ be the quotient maps. 
    
    Then there exist unital completely positive maps $\widetilde{\Phi}_A: M_d(\mathbb{C}) \to \mathcal{A}_{XA}^{POVM}$ and $\widetilde{\Phi}_B: \mathcal{A}_{XA}^{POVM}$ such that $q_A(\widetilde{\Phi}_A(\widetilde{E}_{x,a})) = q_A(e_{x,a})$ and $q_B(\widetilde{\Phi}_B(\widetilde{F}_{y,b})) = q_B(f_{y,b})$ for all $x,y \in X$ and $a,b \in B$. 
\end{lemma}
\begin{proof}
    Define $\pi'_A(a) \coloneq \widetilde{\pi}_A(a) \otimes I_d$ and $\pi'_B(b) = I_d \otimes \widetilde{\pi}_B$. It is then obvious that they are representations of $\mathcal{A}_{XA}^{POVM}$. Since $(\widetilde{\pi}_A \otimes \widetilde{\pi}_B, \ket{\psi})$ is a GNS representation of $\widetilde{\rho}$, one can show that $(\pi'_A, \ket{\widetilde{\psi}})$ and $(\pi'_B, \ket{\widetilde{\psi}})$ are GNS representations of the states $\widetilde{\rho}_A$ and $\widetilde{\rho}_B$ respectively. Hence, by \cref{lem:kernel-GNS-rep}, we see that $\ker(\widetilde{\pi}_A) = \ker{\pi'_A} = \mathcal{J}_A$ and $\ker(\widetilde{\pi}_B) = \ker{\pi'_B} = \mathcal{J}_B$. 
    
    Since $\widetilde{\pi}_A$ and $\widetilde{\pi}_B$ are irreducible, their images are $M_d(\mathbb{C})$. It is now easy to see that there are $*$-isomorphisms $\Phi_A: M_d(\mathbb{C}) \to \mathcal{A}_{XA}^{POVM}/\mathcal{J}_A$ and $\Phi_B: M_d(\mathbb{C}) \to \mathcal{A}_{XA}^{POVM}/\mathcal{J}_B$ respectively. Indeed, these are the inverses of the $\ast$-isomorphisms $\Phi_A': \mathcal{A}_{XA}^{POVM}/\mathcal{J}_A \to M_d(\mathbb{C})$ and $\Phi_B':\mathcal{A}_{XA}^{POVM}/\mathcal{J}_B \to M_d(\mathbb{C})$, which exist because $\widetilde{\pi}_A$ and $\widetilde{\pi}_B$ are surjective. We also note that $\Phi_A(\widetilde{E}_{x,a}) = q(e_{x,a})$ and $\Phi_B(\widetilde{F}_{y,b}) = q(f_{y,b})$ for all $x,y \in X$ and $a,b \in A$. 
    
    By \cite{choi1976completely} there exist unital completely positive maps $\widetilde{\Phi}_A: M_d(\mathbb{C}) \to \mathcal{A}_{XA}^{POVM}$ and and $\widetilde{\Phi}_B: M_d(\mathbb{C}) \to \mathcal{A}_{XA}^{POVM}$ such that $\Phi_A = q_A \circ \widetilde{\Phi}_A $ and $\Phi_B = q_B \circ \widetilde{\Phi}_B$ respectively. Hence, it follows that $q_A(\widetilde{\Phi}_A(\widetilde{E}_{x,a})) = \Phi_A(\widetilde{E}_{x,a}) = q_A(e_{x,a})$ and $q_B(\widetilde{\Phi}_B(\widetilde{F}_{y,b})) = \Phi_B(\widetilde{F}_{y,b}) q_B(e_{y,b})$ for all $x,y \in X$ and $a,b \in B$. 
\end{proof}

\begin{lemma}\label{lem:dilation-measurement-operators}
    Let $(\mathcal{H}_A, \mathcal{H}_B, \{E_{x,a}\}, \{F_{y,b}\}, \ket{\psi})$ be a full-rank strategy for $\widetilde{p}$ and $\pi_A: \mathcal{A}_{XA}^{POVM} \to \mathcal{B(H)}_A$ and $\pi_B: \mathcal{A}_{XA}^{POVM} \to \mathcal{A}_{XA}^{POVM}$ be the associated representations. Then, there exists $s_A, s_B \in \mathbb{N}$ and isometries $V_A: \mathcal{H}_A \to \mathbb{C}^d \otimes \mathbb{C}^{s_A}$ and $V_B: \mathcal{H}_B \to \mathbb{C}^d \otimes \mathbb{C}^{s_B}$ such that $E_{x,a} = V_A^*(\widetilde{E}_{x,a} \otimes I_{s_A})V_A$ and $F_{y,b} = V_B^*(\widetilde{F}_{y,b} \otimes I_{s_B}) V_B$ for all $x,y \in A$ and $a,b \in A$.
\end{lemma}
\begin{proof}
    Consider the UCP maps $\pi_A' \coloneq \pi_A \circ \widetilde{\Phi}_A$ and $\pi_B' \coloneq \pi_B \circ \widetilde{\Phi}_B$ respectively. We now show that $\pi_A'(\widetilde{E}_{x,a}) = E_{x,a}$ and $\pi_B'(\widetilde{F}_{y,b}) = F_{y,b}$ for all $x,y \in X$ and $a,b \in A$. Since $q_A(\widetilde{\Phi}_A(\widetilde{E}_{x,a})) = q(e_{x,a})$ for all $x \in X$ and $a \in A$, we have that $\widetilde{\rho}_A((\widetilde{\Phi}_A(\widetilde{E}_{x,a})-e_{x,a})^*(\widetilde{\Phi}_A(\widetilde{E}_{x,a})-e_{x,a})) = 0$. However, since $\widetilde{\rho}$ is the unique state implementing $\widetilde{p}$, we have 
    \begin{align*}
        & \widetilde{\rho}_A((\widetilde{\Phi}_A(\widetilde{E}_{x,a})-e_{x,a})^*(\widetilde{\Phi}_A(\widetilde{E}_{x,a})-e_{x,a})) \\
        = & \widetilde{\rho}((\widetilde{\Phi}_A(\widetilde{E}_{x,a})-e_{x,a})^*(\widetilde{\Phi}_A(\widetilde{E}_{x,a})-e_{x,a}) \otimes 1) \\
        = &\braket{ \psi \mid (\pi_A'(\widetilde{E}_{x,a}) \otimes I_d -E_{x,a} \otimes I_d)^*(\pi_A'(\widetilde{E}_{x,a}) \otimes I_d - E_{x,a} \otimes I_d)\mid \psi },
    \end{align*}
    so that $\braket{ \psi \mid (\pi_A'(\widetilde{E}_{x,a}) \otimes I_d -E_{x,a} \otimes I_d)^*(\pi_A'(\widetilde{E}_{x,a}) \otimes I_d - E_{x,a} \otimes I_d)\mid \psi } = 0$. This implies that $\pi_A'(\widetilde{E}_{x,a}) = E_{x,a}$ since $\ket{\psi}$ is of full Schmidt-rank. Similarly, we can show that $\pi_B'(\widetilde{F}_{y,b}) = F_{y,b}$ for all $y \in X$ and $ b\in A$. 
    
    Then, it follows from Stinespring's theorem that there are Hilbert spaces $\mathcal{K}_A$ and $\mathcal{K}_B$, $*$-homomorphisms $\pi^{\prime\prime}_A: M_d(\mathbb{C}) \to \mathcal{B}(\mathcal{K}_A)$ and $\pi^{\prime\prime}_B: M_d(\mathbb{C}) \to \mathcal{B}(\mathcal{K}_B)$, and isometries $V_A:\mathcal{H}_A \to \mathcal{K}_A$, $V_A: \mathcal{H}_B \to \mathcal{K}_A$ such that $\pi_A' = V_A^* \pi_A^{\prime\prime} V_A$ and $\pi_B' = V_A^* \pi_A^{\prime\prime} V_B$. Lastly, since $M_d(\mathbb{C})$ has a unique irreducible representation, without loss of generality we may assume that $\mathcal{K}_A = \mathbb{C}^d \otimes \mathbb{C}^{s_A}$, $\mathcal{K}_B = \mathbb{C}^d \otimes \mathbb{C}^{s_B}$; and that $\pi_A(a) = a \otimes I_{s_A}$ and $\pi_B(b) = b \otimes I_{s_B}$. Hence, $E_{x,a} = V_A^*(\widetilde{E}_{x,a} \otimes I_{s_A})V_A$ and $F_{y,b} = V_B^*(\widetilde{F}_{y,b} \otimes I_{s_B}) V_B$ for all $x,y \in A$ and $a,b \in A$.  
\end{proof}

\begin{lemma}\label{lem:unique-eigenvector}
     Let $m = \abs{X}$ and $n = \abs{A}$ and $\widetilde{M} = \sum_{x \in X, a  \in A} \widetilde{E}_{x,a} \otimes \widetilde{F}_{x,a}$. Then, $m$ is the largest eigenvalue of $\widetilde{M}$ and the eigenspace corresponding to $m$ is the one-dimensional space spanned by the maximally entangled state $\ket{\widetilde{\psi}}$. 
\end{lemma}
\begin{proof}
    Let $D = \frac{1}{\sqrt{d}}I_{d}$, so that $\ket{\widetilde{\psi}} = \vect(D)$. Consider $\widetilde{M} \ket{\widetilde{\psi}}$:
    \begin{align*}
        \widetilde{M} \ket{\widetilde{\psi}} & = \sum_{x \in X, a \in A} \widetilde{E}_{x,a} \otimes (\widetilde{E}_{x,a})^T \vect{(D)} \\
        & = \sum_{x \in X, a \in A} \vect(\widetilde{E}_{x,a}D\widetilde{E}_{x,a}) \\
        & =  \frac{1}{\sqrt{d}} \sum_{x \in X, a \in A} \vect{(\widetilde{E}_{x,a})}\\
        & = \frac{1}{\sqrt{d}}\vect{(m I_d)} = m \ket{\widetilde{\psi}}
    \end{align*}
    Hence, $\ket{\widetilde{\psi}}$ is an eigenvector of $M$ with eigenvalue $m$. 
    
    Since $\widetilde{M} \geq 0$, the largest eigenvalue of $\widetilde{M}$ is $\norm{\widetilde{M}}$. Moreover, we have 
    \begin{align*}
        \widetilde{M} & = \sum_{x  \in X, a \in A} (\widetilde{E}_{x,a} \otimes \widetilde{F}_{x,a}) \\
        & \leq \sum_{x  \in X, a \in A} \widetilde{E}_{x,a} \otimes I_d \\
        & = m (I_d \otimes I_d),
    \end{align*}

    which shows that $\norm{\widetilde{M}} \leq m$. Hence, $\norm{\widetilde{M}} = m$ and $m$ is the largest eigenvalue of $\widetilde{M}$. 

    In order to show that the eigenspace corresponding to $m$ is one-dimensional, let $\ket{\varphi} \in \mathbb{C}^d \otimes \mathbb{C}^d$ be a unit vector such that $\widetilde{M} \ket{\varphi} = m \ket{\varphi}$. Let $B = \mat{\ket{\varphi}}$. Then, $\widetilde{M} \ket{\varphi} = m \ket{\varphi}$ is equivalent to saying $\sum_{x \in X, a \in A} \widetilde{E}_{x,a}B \widetilde{E}_{x,a} = m B$. Consider the quantum channel 
    $$X \mapsto \frac{1}{m} \sum_{x \in X, a \in A} \widetilde{E}_{x,a} X \widetilde{E}_{x,a}$$ for $X \in M_d(\mathbb{C})$. 
    Since $B$ is a fixed point of this quantum channel, it follows from \cite{watrous2018theory}*{Theorem 4.25} that $B$ commutes with each $\widetilde{E}_{x,a}$. However, since the representation $\widetilde{\pi}$ is irreducible, the projections $\widetilde{E}_{x,a}$ generate all of $M_d(\mathbb{C})$. Hence, $B$ commutes with all elements in $M_d(\mathbb{C})$, which implies that $B = \lambda I_d$ for some $\lambda \in \mathbb{C}$, which shows that $\ket{\varphi} = \ket{\widetilde{\psi}}$ as required.  
\end{proof}

We shall prove a robust generalisation of the following result in \cref{lem:local-eps-dil-implies-closeness-of-state}. For now we give a version that follows from the results of \cite{paddock_operator-algebraic_2023}. 

\begin{lemma}\label{lem:local-dilation-implies-equal-state}
     Let $S = (\mathcal{H}_A, \mathcal{H}_B, \{E_{x,a}\}, \{F_{y,b}\}, \ket{\psi})$ be a strategy for $\widetilde{p}$ such that $\widetilde{S}$ is a local dilation of $S$. Then, the state on $\mathcal{A}_{XA}^{POVM} \otimes \mathcal{A}_{XA}^{POVM}$ induced by $S$ is $\widetilde{\rho}$. 
\end{lemma}
\begin{proof}
    Follows from \cite{paddock_operator-algebraic_2023}*{Proposition 4.8}.
\end{proof}

We shall now show the non-robust part of \cref{thm:main-result}. We shall be showing this under a weaker assumption on the correlation $\widetilde{p}$. In particular, we only assume that there is a unique state $\rho$ in $S_{\fin}(\mathcal{A}_{XA}^{POVM} \otimes \mathcal{A}_{XA}^{POVM})$ such that $p(a,b \mid x,y) = \rho(e_{x,a} \otimes e_{y,b})$ for all $x,y \in X$ and $a,b \in A$, instead of assuming that there is only one such state in all of $S(\mathcal{A}_{XA}^{POVM} \otimes \mathcal{A}_{XA}^{POVM})$. 

\begin{theorem}\label{thm:main-theorem-1}
    Let $\widetilde{p} = \{\widetilde{p}(a,b\mid x,y)\}_{x, y \in X, a,b \in A}$ be a synchronous quantum correlation. Then, the following are equivalent: 
    \begin{enumerate}[label = \roman*.]
        \item $\widetilde{p}$ satisfies condition~\ref{self-test-condition}
        \item The correlation $\widetilde{p}$ self-tests a PME strategy $(\mathbb{C}^d, \mathbb{C}^d, \{\widetilde{E}_{x,a}\}, \{\widetilde{F}_{x,a}\}, \ket{\widetilde{\psi}} )$.
    \end{enumerate}
\end{theorem}
\begin{proof}
    First, we shall assume that $\widetilde{p}$ satisfies condition~\ref{self-test-condition} and show that $\widetilde{p}$ is a full-rank self-test for the PME strategy $(\mathbb{C}^d, \mathbb{C}^d, \{\widetilde{E}_{x,a}\}, \{\widetilde{F}_{x,a}\}, \ket{\widetilde{\psi}} )$. It then follows that $\widetilde{p}$ is an assumption-free self-test by \cite{baptista_mathematical_2023}*{Theorem 4.1}. 

    Let $S = (\mathbb{C}^{r_A}, \mathbb{C}^{r_B}, \{E_{x,a}\}, \{F_{y,b}\}, \ket{\psi})$ be a full-rank strategy for $\widetilde{p}$. By our assumption and \cref{lem:dilation-measurement-operators}, we have that there are isometries $V_A: \mathbb{C}^{r_A} \to \mathbb{C}^d \otimes \mathbb{C}^{s_A}$ and $V_B: \mathbb{C}^{r_B} \to \mathbb{C}^d \otimes \mathbb{C}^{s_B}$ such that $ E_{x,a} = V_A^*(\widetilde{E}_{x,a} \otimes I_{s_A}) V_A$ and $F_{y,b} = V_B^*(\widetilde{F}_{y,b} \otimes I_{s_B}) V_B$ for all $x, y \in X$ and $a,b \in A$. Let $V \coloneq V_A \otimes V_B$ and note that 
    \begin{align*}
        m = \sum_{x \in X, a \in A} \braket{\psi \mid E_{x,a} \otimes F_{y,b} \mid \psi } & = \sum_{x \in X, a \in A} \braket{V \psi \mid (\widetilde{E}_{x,a} \otimes I_{s_A}) \otimes (\widetilde{F}_{y,b}\otimes I_{s_B}) \mid V \psi } \\
        & = \braket{V \psi \mid \widetilde{M} \otimes I_{s_A} \otimes I_{s_B} \mid V \psi},
    \end{align*}
    which shows that $V\ket{\psi}$ is an eigenvector of $\widetilde{M} \otimes I_{s_A} \otimes I_{s_B}$ with eigenvalue $m$, where $\widetilde{M}$ is as defined in \cref{lem:unique-eigenvector}. It follows from \cref{lem:unique-eigenvector} that $\ket{\widetilde{\psi}}$ spans the eigenspace of $\widetilde{M}$ corresponding to $m$. Hence, we have that $V \ket{\psi} = \ket{\widetilde{\psi}}  \otimes \ket{\psi_{aux}}$ for some unit vector $\ket{\psi_{aux}} \in \mathbb{C}^d \otimes \mathbb{C}^d$. This shows that $\widetilde{S}$ is a local dilation of $S$ which finishes one direction of the proof. 

    For the other direction, let us assume that there is a PME strategy $\widetilde{S}$ that is self-tested by $\widetilde{p}$ and let $\widetilde{\rho}$ be the state induced $\widetilde{S}$ on $\mathcal{A}_{XA}^{POVM} \otimes \mathcal{A}_{XA}^{POVM}$. Let $S$ be a strategy that induces $\widetilde{p}$. Then, $\widetilde{S}$ is a local dilation of $S$, which together with \cref{lem:local-dilation-implies-equal-state} implies that the state induced by $S$ is $\widetilde{\rho}$. This finishes the proof. 
\end{proof}

\section{Characterising Robust Synchronous Self-Tests}

In this section, we shall establish \cref{thm:main-result}. We begin this section by showing that $S(\mathcal{A}_{XA}^{POVM} \otimes_{\min} \mathcal{A}_{YB}^{POVM})$ is the closure of $S_{\fin}(\mathcal{A}_{XA}^{POVM} \otimes_{\min} \mathcal{A}_{YB}^{POVM})$. This implies that $\mathcal{A}_{XA}^{POVM} \otimes_{\min} \mathcal{A}_{YB}^{POVM}$ is residually finite dimensional, i.e. it has a separating family of finite-dimensional representations. 
\begin{prop}\label{prop:state-correlation-correspondence-qa}
    $S(\mathcal{A}_{XA}^{POVM} \otimes_{\min} \mathcal{A}_{YB}^{POVM})$ is the weak-$*$ closure of $S_{\fin}(\mathcal{A}_{XA}^{POVM} \otimes_{\min} \mathcal{A}_{YB}^{POVM})$. In particular, $p \in C_{qa}(X,Y,A,B)$ if and only if there is a state $\rho$ on $\mathcal{A}_{XA}^{POVM} \otimes_{\min} \mathcal{A}_{YB}^{POVM}$ such that $p(a,b \mid x,y) = \rho(e_{x,a} \otimes e_{y,b})$ for all $x \in X, y \in Y$, $a \in A ,b \in A$. 
\end{prop}
\begin{proof}
    We begin by noting that since $\mathcal{A}_{XA}^{POVM}$ is separable, the unit ball of $\mathcal{A}_{XA}^{POVM}$ is also separable. If $\Omega$ is some weak-$*$ dense countable set of states in $\mathcal{A}_{XA}^{POVM}$ and $\{\Pi_\omega\}_{\omega \in \Omega}$ is the set of GNS-representations of states in $\Omega$, then the direct sum of these representations gives an isometric representation of $\mathcal{A}_{XA}^{POVM}$. Since $\mathcal{A}_{XA}^{POVM}$ is separable, each of these representations is also separable, so that $\mathcal{A}_{XA}^{POVM}$ has a separable isometric representation, say $\Pi_A: \mathcal{A}_{XA}^{POVM}\to \mathcal{B(H}_A)$.
    
    Hence, there is an isometric representation $\Pi_A \otimes \Pi_B : \mathcal{A}_{XA}^{POVM} \otimes \mathcal{A}_{YB}^{POVM} \to \mathcal{B(H}_A) \otimes \mathcal{B(H}_A)$, so that we may identify $\mathcal{A}_{XA}^{POVM} \otimes \mathcal{A}_{YB}^{POVM}$ as a $C^*$-subalgebra of $\mathcal{B(H}_A) \otimes \mathcal{B(H}_B)$. It follows from standard results (see \cite{kadison1997elementary}*{Corollary 4.3.10} for example) that the set of states on $\mathcal{B(H}_A) \otimes \mathcal{B(H}_B) \cong \mathcal{B}(\mathcal{H}_A \otimes \mathcal{H}_B)$ is equal to the closed convex hull of the set of vector states on $ \mathcal{B}(\mathcal{H}_A \otimes \mathcal{H}_B)$, each of which is a tensor product state on $\mathcal{A}_{XA}^{POVM} \otimes_{\min} \mathcal{A}_{XA}^{POVM}$. 

     Now, let $\ket{\psi} \in \mathcal{H}_A \otimes \mathcal{H}_B$, $\rho$ be the state $a \to \braket{\psi \mid \Pi_A \otimes \Pi_B(a) \mid \psi}$ and choose a basis $\{e_1, e_2, \dots\}$ for $\mathcal{H}_A$ and $\{f_1, f_2, \dots\}$ for $\mathcal{H}_B$. Let $\pi_n^A$ be the projection onto the span of $\{e_1, \dots, e_n\}$, $\pi_n^B$ be the projection on to the span of $\{f_1, \dots, f_n\}$ and let $\ket{\psi_n} = \pi_n^A \otimes \pi_n^A (\ket{\psi})$. If we define $\rho_n(a) = \braket{ \psi_N\mid (\pi_n^A \Pi_A \pi_n^A \otimes \pi_n^B \Pi_B \pi_n^B) (a) \mid \psi_N}$, then it is easy to show that $\rho_N \to \rho$. Hence, the set of finite-dimensional states is dense in the set of tensor-product states on $\mathcal{A}_{XA}^{POVM} \otimes \mathcal{A}_{YB}^{POVM}$, which shows that $S(\mathcal{A}_{XA}^{POVM} \otimes \mathcal{A}_{YB}^{POVM})$ is the weak-$*$ closure of $S_{\fin}(\mathcal{A}_{XA}^{POVM} \otimes \mathcal{A}_{YB}^{POVM})$. 

    It is evident from the preceding discussion that if $\rho \in S(\mathcal{A}_{XA}^{POVM} \otimes \mathcal{A}_{YB}^{POVM})$ then the correlation $(\rho(e_{x,a} \otimes f_{y,b}))_{x,y \in X, a,b \in A}$ is in $C_{qa}(X,Y,A,B)$. We shall now show the converse. Let $p \in C_{qa}(X,X,A,A)$ then, there are finite-dimensional quantum correlations $p_k$ such that $p_k \to p$. Let $S_k$ be a finite-dimensional model for each $p_k$ and let $\rho_k$ be the state associated with $S_k$. It is then obvious that if we set $\rho$ to be a weak-$*$ limit point of $\{\rho_k\}_{k \in \mathbb{N}}$ then $\rho(e_{x,a} \otimes f_{y,b}) = \lim_{k \to \infty} \rho_k(e_{x,a} \otimes f_{y,b}) = p(a,b \mid x,y).$ 
\end{proof}

\subsection{Analogue of Gowers-Hatami}

In this subsection, we show that given a synchronous correlation $\widetilde{p}$ that is an extreme point  of $C_q(X, X, A, A)$ satisfying the hypothesis~\cref{robust-self-test-condition}, any finite-dimensional quantum strategy implementing a correlation close to $\widetilde{p}$ has measurement operators that are close to local dilations of the honest strategy. We generalise the techniques introduced in \cite{mancinska_constant-sized_2024} and make use of several lemmas introduced in the previous section. 

\begin{theorem}\label{thm:gow-hat}
    Let $\widetilde{p}$ be a synchronous quantum correlation that is an extreme point of $C_q(X, X, A, A)$ satisfying condition~\ref{self-test-condition} and let $\widetilde{S}$ be the PME strategy for $\widetilde{p}$ from \cref{lem:canonical-PME-strat}. Then, for all $\epsilon > 0$, there exists a $\delta > 0$ such that for all $r_A, r_B \in \mathbb{N}$, density matrices $\rho \in M_{r_Ar_B}(\mathbb{C})$ and POVMs $\{E_{x,a}\} \subset M_{r_A}(\mathbb{C}), \ \{F_{y,b}\} \subset M_{r_B}(\mathbb{C})$ satisfying 
    $$\sum_{x,y \in X, a,b \in A} \abs{\Tr_{\rho}(E_{x,a} \otimes F_{y,b}) - \widetilde{p}(a,b\mid x,y)} \leq \delta,$$
    there are isometries $V_A: \mathbb{C}^{r_A} \to \mathbb{C}^{d}\otimes \mathbb{C}^{s_A}$ and $V_B: \mathbb{C}^{r_B} \to \mathbb{C}^d \otimes \mathbb{C}^{s_B}$, for some $s_A, s_B \in \mathbb{N}$, such that 
    \[\norm{E_{x,a} - V_A^*(\widetilde{E}_{x,a} \otimes I_{s_A})V_A }_{\rho_A} \leq \epsilon \text{ and } \norm{F_{y,b} - V_B^*(\widetilde{F}_{y,b} \otimes I_{s_B})V_B}_{\rho_B} \leq \epsilon\]
    for all $x,y \in X$ and $a,b \in A$. 
\end{theorem}
\begin{proof}
    We proceed by contradiction. Let us fix a counter example for the result, i.e., an $\epsilon > 0$ and for each $k \geq 1$, $\delta_k \geq 0$ such that $\lim_{k \to \infty} \delta_k = 0$, density matrices $\rho^{[k]} \in M_{r_A^{[k]}r_B^{[k]}}(\mathbb{C})$ and POVMs $\{E_{x,a}^{[k]}\} \subseteq M_{r_A^{[k]}}(\mathbb{C})$, $\{F_{y,b}^{[k]}\} \subseteq M_{r_B^{[k]}}(\mathbb{C})$ such that 
   
   $$\sum_{x, y \in X, a, b \in A} \abs{\Tr_{\rho^{[k]}}(E_{x,a} \otimes F_{y,b}) - \widetilde{p}(a,b\mid x,y)} \leq \delta_k,$$ 
   and for all isometries $V_A : \mathbb{C}^{r_A^{[k]}} \to \mathbb{C}^d \otimes \mathbb{C}^{s_A}$ and $V_B : \mathbb{C}^{r_B^{[k]}} \to \mathbb{C}^d \otimes \mathbb{C}^{s_B}$, there exists an $x \in X$ and $a \in A$ or $y \in Y$ and $b \in B$ such that $$\norm{E_{x,a}^{[k]} - V_a^* (\widetilde{E}_{x,a} \otimes I_{s_A}) V_a }_{\rho^{[k]}_A} > \epsilon \text{ or } \norm{F_{y,b}^{[k]} - V_B^* (\widetilde{F}_{y,b} \otimes I_{s_B}) V_B}_{\rho^{[k]}_B} > \epsilon$$  

   Let $\pi_A^{[k]} \otimes \pi_B^{[k]} \colon \mathcal{A}_{XA}^{POVM} \otimes \mathcal{A}_{XA}^{POVM} \to M_{r_A^{[k]}}(\mathbb{C}) M_{r_B^{[k]}}(\mathbb{C})$ be the associated representation with the POVMs $\{E_{x,a}^{[k]}\}, \ \{F_{y,b}^{[k]}\}$ and let $\widetilde{\rho}_k \coloneq \Tr_{\rho^{[k]}} \circ (\pi_A^{[k]} \otimes \pi_B^{[k]})$ be states induced on $\mathcal{A}_{XA}^{POVM} \otimes \mathcal{A}_{XA}^{POVM}$. We see that any weak-$*$ limit point $\rho$ of these state (which exists by virtue of the Banach-Alagolu theorem) satisfies $\rho(e_{x,a} \otimes e_{y,b}) = \widetilde{p}(a,b \mid x, y)$ for all $x,y \in X$ and $a,b \in A$. It then follows by our assumption that $\rho = \widetilde{\rho}$. Hence, by passing to a sub-sequence if necessary, we may assume that $\widetilde{\rho}_k \to \widetilde{\rho}$. 

   Let $\widetilde{\rho}_A$ and $\mathcal{J}_A$ be as defined in \cref{lem:UCP-map-from-canon-strat}. Then, one has that $a \in \mathcal{J}_A$ if and only if $\lim_{k \to \infty}\norm{\pi_A^{[k]}(a)}_{\rho^{[k]}_A} = 0$. Let $\widetilde{\Phi}_A: M_d(\mathbb{C}) \to \mathcal{A}_{XA}^{POVM}$ be the UCP map from \cref{lem:UCP-map-from-canon-strat} and define $\widetilde{\pi}_A^{[k]} = \pi_A^{[k]} \circ \widetilde{\Phi}_A$. In particular, one has that $\widetilde{\Phi}_A(\widetilde{E}_{x,a}) - e_{x,a}\in \mathcal{J}_A$, so that 
   
   \begin{align*}
       \lim_{k \to \infty} \norm{\pi_A^{[k]}(\widetilde{\Phi}_A(\widetilde{E}_{x,a}) - e_{x,a})}_{\rho^{[k]}_A} = \lim_{k \to \infty} \norm{\widetilde{\pi}_A^{[k]}(\widetilde{E}_{x,a}) - E_{x,a}^{[k]}}_{\rho_A^{[k]}} = 0
   \end{align*} for all $x \in X$ and $a \in A$. 

   Since each $\widetilde{\pi}_A^{[k]}$ is a UCP map, using Stinespring's theorem and uniqueness of irreducible representation of $M_d(\mathbb{C})$ in the same fashion as in the proof of \cref{lem:dilation-measurement-operators}, we have that there are isometries $V_A^{[k]}: \mathbb{C}^{r_A^{[k]}} \to \mathbb{C}^d \otimes \mathbb{C}^{s_A^{[k]}}$ such that $\widetilde{\pi}_A^{[k]}(a) = (V_A^{[k]})^*(a \otimes I_{s_A^{[k]}})V_A^{[k]}$. 
   
   In particular, for some large $k$ we have $\norm{(V_A^{[k]})^*(\widetilde{E}_{x,a} \otimes I_{s_A^{[k]}})V_A^{[k]} - E_{x,a}^{[k]}}_{\rho_k} \leq \epsilon$ for all $x \in X$ and $a \in A$. Repeating the same argument for $\{\widetilde{F}_{y,b}\}$ we get a contradiction which finishes the proof. 
\end{proof}

\subsection{Necessary and Sufficient Conditions for Robust Self-testing}

In order to finish this proof of the forward direction in \cref{thm:main-result}, we need to borrow a few lemmas from \cite{mancinska_constant-sized_2024}:  

\begin{lemma}[\cite{mancinska_constant-sized_2024}*{Lemma 6.7}]\label{lem:borr-lem-1}
    Let $A \in M_d(\mathbb{C})$ be a positive matrix having at least two eigenvalues and let $\lambda_1 > \lambda_2 > \dots \lambda_l \geq 0$ be the distinct eigenvalues of $A$. Let $\xi \in \mathbb{C}^d$ be a unit vector such that $\lambda_1 - \epsilon \leq \braket{A \xi, \xi}$ for some $\epsilon \geq 0$. If $Q_1$ is the projection onto the eigenspace corresponding to $\lambda_1$, then
    \[\norm{Q_1\xi}^2 \geq 1 - \frac{\epsilon}{\lambda_1 - \lambda_2}\]
\end{lemma}

\begin{lemma}[\cite{mancinska_constant-sized_2024}*{Lemma 6.8}]\label{lem:borr-lem-2}
     Let $d = d_A d_B$, $\psi \in \mathbb{C}^d$ be a unit vector, $X_1, X_2 \in M_{d_A}(\mathbb{C})$, and $Y_1, Y_2 \in M_{d_B}(\mathbb{C})$. Let $\rho_A, \rho_B$ be the reduced density matrices of the density matrix $\rho = \ket{\psi}\bra{\psi}$. Then, 
     \begin{align}
         \abs{\braket{\psi \mid (X_1 \otimes Y_1 - X_2 \otimes Y_2) \mid \psi}} \leq \norm{X_1 - X_2}_{\rho_A}\norm{Y_2^*}_{\rho_B} + \norm{Y_1 - Y_2}_{\rho_B}\norm{X_1^*}_{\rho_A}\text{, and}
    \end{align}
    \begin{align}
         \norm{(X_1 \otimes Y_1 - X_2 \otimes Y_2) \ket{\psi}} \leq \norm{X_1 - X_2}_{\rho_A}\norm{Y_2} + \norm{Y_1 - Y_2}_{\rho_B}\norm{X_1}
     \end{align}
\end{lemma}

\begin{lemma}[\cite{mancinska_constant-sized_2024}*{Lemma 6.9}]\label{lem:borr-lem-3}
    Let $\xi, \eta \in \mathbb{C}^d$ and let $P \in M_d(\mathbb{C})$ be a projection. Suppose $\abs{\norm{\xi}^2 - \norm{\eta}^2} \leq \epsilon_1$ and $\norm{\xi - P \eta} \leq \epsilon_2$ for some $\epsilon_1$, $\epsilon_2 > 0$. Then, 
    $$\norm{\xi - \eta} \leq \epsilon_2 + \sqrt{\epsilon_1 + (\norm{\xi} + \norm{\eta})\epsilon_2}$$
\end{lemma}
 We now prove the forward direction of the proof of \cref{thm:main-result}.

\begin{theorem}\label{thm:main-theorem-2}
    Let $\widetilde{p} \in \partial C_q(X,X,A,A) \cap C_q^s(X,A)$ be a synchronous quantum correlation satisfying condition~\ref{robust-self-test-condition}. Then, there is a PME strategy $\widetilde{S} = \{\mathbb{C}^d, \mathbb{C}^d, \{\widetilde{E}_{x,a}\}, \{\widetilde{F}_{y,b}\}, \ket{\widetilde{\psi}}\}$ such that $\widetilde{p}$ robustly self-tests $\widetilde{S}$.  
\end{theorem}
\begin{proof}
    We shall show that $\widetilde{p}$ is a robust self-test for all projective models. It shall follow from \cite{baptista_mathematical_2023}*{theorem 4.1} that $\widetilde{p}$ is an assumption-free self-test. Fix an $\epsilon > 0$ and let $\widetilde{p}$ be as in the statement of the theorem and $\widetilde{S}$ be the strategy for $\widetilde{p}$ from \cref{lem:canonical-PME-strat}. Since condition~\ref{robust-self-test-condition} is stronger than condition~\ref{self-test-condition}, it follows from \cref{thm:main-theorem-1} that $\widetilde{p}$ is a self-test for $\widetilde{S}$. We now work towards showing that it is a robust self-test. Let $\widetilde{M}$ be as defined in \cref{lem:unique-eigenvector} and let $\lambda_2$ be the second largest eigenvalue of $\widetilde{M}$. 
    
    Choose $\epsilon' > 0$ such that $\epsilon' < \frac{m-\lambda_2}{2mn+1}$ and $\epsilon \geq 2\epsilon' + \beta + \sqrt{5\epsilon'+2\beta}$, where $\beta = \sqrt{\frac{2(2mn+1)\epsilon'}{m-\lambda_2}}$. Now, apply \cref{thm:gow-hat} to obtain a $\delta' > 0$ satisfying the hypothesis of \cref{thm:gow-hat} for $\epsilon'$. Hence, if $p$ is a quantum correlation with model $S = (\mathbb{C}^{r_A}, \mathbb{C}^{r_B}, \{E_{x,a}\}, \{F_{y,b}\}, \ket{\psi})$ such that $\norm{p - \widetilde{p}}_1 \leq \delta'$, there are isometries $V_A: \mathbb{C}^{r_A} \to \mathbb{C}^d \otimes \mathbb{C}^{s_A}$ and $V_B: \mathbb{C}^{r_B} \to \mathbb{C}^d \otimes \mathbb{C}^{s_B}$ such that 
    \[\norm{E_{x,a} - V_A^*(\widetilde{E}_{x,a} \otimes I_{s_A})V_A }_{\rho_A} \leq \epsilon' \text{ and } \norm{F_{y,b} - V_B^*(\widetilde{F}_{y,b} \otimes I_{s_B})V_B}_{\rho_B} \leq \epsilon' \]
    for all $x,y \in X$ and $a,b \in A$. Set $\delta = \min\{\delta', \epsilon'\}$. We shall show that $\widetilde{S}$ is a local $\epsilon$-dilation of $\widetilde{S}$ whenever $\norm{p - \widetilde{p}}_1 \leq \delta$. 
    
    Since $\delta < \delta'$, by \cref{lem:borr-lem-1} we have the following: 
    \begin{align}\label{eq:closeness-of-measurements}
    \begin{aligned}
        \abs{\braket{\psi \mid (E_{x,a} \otimes F_{y,b}) - (V_A \otimes V_B)^* (\widetilde{E}_{x,a} \otimes I_{s_A} \otimes \widetilde{F}_{y,b} \otimes I_{s_B}) (V_A \otimes V_B) \mid \psi}} \\
        \leq \norm{E_{x,a} - V_A^*(\widetilde{E}_{x,a} \otimes I_{s_A})V_A }_{\rho_A} + \norm{F_{y,b} - V_B^*(\widetilde{F}_{y,b} \otimes I_{s_B})V_B}_{\rho_B} \leq 2\epsilon' \leq \epsilon
    \end{aligned}  
    \end{align}

    We now work towards constructing a state $\ket{\aux} \in \mathbb{C}^{s_A} \otimes \mathbb{C}^{s_B}$ such that $(V_A \otimes V_B)\ket{\psi} \approx_{\epsilon} \ket{\widetilde{\psi}}\ket{\aux}$. Define $M \coloneq \sum_{x,a} E_{x,a} \otimes F_{x,a}$. Substituting $x,a = y,b$ in the above equation and summing over all $x,a$, we get:
    \[\abs{\braket{\psi \mid M \mid \psi} - \braket{\psi \mid (V_A \otimes V_B)^* (\widetilde{M} \otimes I_{s_A} \otimes I_{s_B}) (V_A \otimes V_B) \mid \psi}} \leq 2mn\epsilon \]
    Note that $\braket{\psi \mid M \mid \psi} = \sum_{x,a} p(a,a \mid x,x) \leq m$. Moreover, since $\norm{p - \widetilde{p}}_1 \leq \delta \leq \epsilon'$ and $\sum_{x,a} \widetilde{p}(a,a \mid x,x) = m$, we have $ \braket{\psi \mid \mid M \psi} \geq  m-\epsilon$. Hence, we can now write   
    \[\braket{\psi \mid(V_A \otimes V_B)^* (\widetilde{M} \otimes I_{s_A} \otimes I_{s_B}) (V_A \otimes V_B) \mid \psi} \geq m - (2mn+1)\epsilon'\]

    Let us denote $(V_A \otimes V_B)\ket{\psi}$ by $\ket{\psi'}$ for convenience. Recall that $m$ is the largest eigenvalue of $\widetilde{M}$ and hence also of $\widetilde{M} \otimes I_{s_A} \otimes I_{s_B}$. Let $Q$ be the projection onto the eigenspace of $\widetilde{M} \otimes I_{s_A} \otimes I_{s_B}$ corresponding to $m$, and let 
    $ \alpha\coloneq \norm{Q \ket{\psi'}}$. Then, $\alpha \geq \left(1 - \frac{(2mn+1)\epsilon'}{m - \lambda_2}\right)^{\frac{1}{2}}$, by \cref{lem:borr-lem-1} and the previous inequality. Here, it should also be noted that $\left(1 - \frac{(2mn+1)\epsilon'}{m - \lambda_2}\right)^{\frac{1}{2}} > 0$ by the choice of $\epsilon'$, so that $\alpha> 0$. 

    Hence, $\alpha^{-1} Q \ket{\psi'}$ is a unit vector in the eigenspace of $\widetilde{M} \otimes I_{s_A} \otimes I_{s_B}$ corresponding to the eignvalue $m$ and there exists a unique quantum state $\ket{\aux} \in \mathbb{C}^{s_A} \otimes \mathbb{C}^{s_B}$ such that $Q \ket{\psi'} = \alpha \ket{\widetilde{\psi}}\ket{\aux}$. 

    Now, we note that 
    \begin{align*}
    \norm{\ \ket{\psi'} - \ket{\widetilde{\psi}}\ket{\aux}\ }^2 & = \norm{\ \ket{\psi'} - \alpha^{-1} Q\ket{\psi'}\ }^2 \\
    & = \braket{\psi'\mid \psi'} - 2\alpha^{-1} \braket{\psi' \mid Q \mid \psi'} + \alpha^{-2} \braket{\psi'\mid Q^2\mid \psi'} \\
    & = 1 - 2 \alpha^{-1}\alpha^2 + \alpha^{-2}\alpha^2 \\
    & = 2(1 - \alpha) \leq \sqrt{2(1-\alpha^2)} \ (\because \alpha \leq 1)
    \end{align*}

    Hence, We have the following: 
    \begin{equation}\label{eq:closeness-of-state}
        \norm{\ \ket{\psi'} - \ket{\widetilde{\psi}}\ket{\aux} \ } \leq \sqrt{2(1-\alpha^2)} \leq \sqrt{\frac{2(2mn+1)\epsilon'}{m-\lambda_2}} = \beta \leq \epsilon
    \end{equation}

    We now work towards showing that 
    $$V_A \otimes V_B (E_{x,a} \otimes F_{y,b}) \ket{\psi} \approx_{\epsilon} (\widetilde{E}_{x,a} \otimes \widetilde{F}_{y,b}) \ket{\widetilde{\psi}} \otimes \ket{\aux}$$ for all $x, y \in X$ and $ a, b \in A$. Fix $V = V_A \otimes V_B$. Then, we haeve 
    \begin{align*}
         \begin{split} \norm{V^*(\widetilde{E}_{x,a} \otimes I_{s_A} \otimes \widetilde{F}_{y,b} \otimes I_{s_B}) V \ket{\psi} - (E_{x,a} \otimes F_{y,b}) \ket{\psi}} \leq & \norm{V_A^* (\widetilde{E}_{x,a} \otimes I_{s_A}) V_A- E_{x,a}}_{\rho_A}\norm{F_{y,b}} + \\ & \norm{V_B^* (\widetilde{F}_{y,b} \otimes I_{s_B})V_B - F_{y,b}}_{\rho_B} \norm{\widetilde{E}_{x,a} \otimes I_{s_A}}\\ & \leq 2 \epsilon'.
         \end{split}
    \end{align*}

    Since $V$ is an isometry, we get $$\norm{V V^*(\widetilde{E}_{x,a} \otimes I_{s_A} \otimes \widetilde{F}_{y,b} \otimes I_{s_B}) V \ket{\psi} - V (E_{x,a} \otimes F_{y,b}) \ket{\psi}} \leq 2 \epsilon'.$$ It follows from \cref{eq:closeness-of-state} that  
    $$\norm{V V^*(\widetilde{E}_{x,a} \otimes I_{s_A} \otimes \widetilde{F}_{y,b} \otimes I_{s_B}) \ket{\widetilde{\psi}}\ket{\aux} - V V^*(\widetilde{E}_{x,a} \otimes I_{s_A} \otimes \widetilde{F}_{y,b} \otimes I_{s_B})V \ket{\psi} } \leq  \norm{\ \ket{\psi'} - \ket{\widetilde{\psi}}\ket{\aux} \ } \leq \beta$$

    Hence, by applying the triangle equality we get
    \begin{align}\label{eq:useful-eq-1}
        \norm{V(E_{x,a} \otimes F_{y,b}) \ket{\psi} - V V^*(\widetilde{E}_{x,a} \otimes I_{s_A} \otimes \widetilde{F}_{y,b} \otimes I_{s_B}) \ket{\widetilde{\psi}}\ket{\aux}} \leq 2 \epsilon' + \beta.
    \end{align}

    Note that $\norm{(\widetilde{E}_{x,a} \otimes I_{s_A} \otimes \widetilde{F}_{y,b} \otimes I_{s_B})\ket{\widetilde{\psi}}\ket{\psi}}^2 = \widetilde{p}(a,b \mid x, y)$. Similarly, $p(a,b \mid x,y ) = \norm{(E_{x,a} \otimes F_{y,b})\ket{\psi}}^2$. Since $V$ is an isometry and $\norm{p - \widetilde{p}}_1 \leq \delta \leq \epsilon'$, we get 
    \begin{equation}\label{eq:useful-eq-2}
    \abs{\ \norm{(\widetilde{E}_{x,a} \otimes I_{s_A} \otimes \widetilde{F}_{y,b} \otimes I_{s_B})\ket{\widetilde{\psi}}\ket{\psi}}^2- \norm{V(E_{x,a} \otimes F_{y,b})\ket{\psi}}^2\ } \leq \epsilon'    
    \end{equation}

    Since $V$ is an isometry, $VV^*$ is a projection. Hence, it follows from \cref{eq:useful-eq-1}, \cref{eq:useful-eq-2} and \cref{lem:borr-lem-3} that 

    \begin{equation}\label{eq:closeness-for-meas-act-on-state}
        \norm{V (E_{x,a} \otimes F_{y,b}) \ket{\psi} - (\widetilde{E}_{x,a} \otimes \widetilde{F}_{y,b}) \ket{\widetilde{\psi}}\ket{\aux}} \leq 2 \epsilon' + \beta + \sqrt{\epsilon'+2(2\epsilon'+\beta)} \leq \epsilon
    \end{equation}

    It follows from \cref{eq:closeness-for-meas-act-on-state} and \cref{eq:closeness-of-state} that $\widetilde{S}$ is a local $\epsilon$-dilation of $S$ which finishes the proof. 
\end{proof}

In particular, one has the following sufficient condition for robust self-testing: 

\begin{cor}
    Let $\widetilde{p} \in C_q(X,X,A,A)$ be a synchronous quantum correlation that is an extreme point of $C_q(X,X,A,A)$. If there is a unique (finite-dimensional) state $\widetilde{\rho}$ on $\mathcal{A}_{XA}^{POVM} \otimes_{\max} \mathcal{A}_{XA}^{POVM}$ such that $\widetilde{\rho}(e_{x,a} \otimes f_{y,b}) = \widetilde{p}(a,b\mid x,y)$ for all $x,y \in X$ and $a,b \in A$, then $\widetilde{p}$ robustly self-tests a PME strategy $\widetilde{S}$. 
\end{cor}
This result can be thought of as a robust analogue of \cite{paddock_operator-algebraic_2023}*{Theorem 3.5} for synchronous correlations.

We now work towards showing the other implication of \cref{thm:main-result}. We begin by that states induced by local $\epsilon$-dilations of $\widetilde{S}$ are close to $\widetilde{\rho}$ on monomials of small length. 

\begin{lemma}\label{lem:local-eps-dil-implies-closeness-of-state}
    Let $S = (\mathcal{H}_A, \mathcal{H}_B, \{E_{x,a}\},\{F_{y,b}\}, \ket{\psi})$ finite dimensional strategy that is a local $\epsilon$-dilation of a PME strategy $\widetilde{S} = (\widetilde{\mathcal{H}}_A, \widetilde{\mathcal{H}}_B, \{\widetilde{E}_{x,a}\},\{\widetilde{F}_{y,b}\}, \ket{\widetilde{\psi}})$. Let $\alpha, \beta \in \mathcal{A}_{XA}^{POVM}$ be monomials of length $l(\alpha)$ and $l(\beta)$ in $\{e_{x,a}\}$. If $\rho$ and $\widetilde{\rho}$ are the states on $\mathcal{A}_{XA}^{POVM} \otimes \mathcal{A}_{XA}^{POVM}$ induced by $S$ and $\widetilde{S}$ respectively, then one has $$\abs{\widetilde{\rho}(\alpha \otimes \beta) - \rho(\alpha \otimes \beta)} \leq ((n+5)(l(\alpha) + l(\beta)) +2)\epsilon.$$      
\end{lemma}
\begin{proof}
Since $\widetilde{S}$ is a local $\epsilon$-dilation of $S$, there exist Hilbert spaces $\mathcal{H}^{\aux}_A$, $\mathcal{H}^{\aux}_B$, isometries $V_A: \mathcal{H}_A \to \widetilde{\mathcal{H}}_A \otimes \mathcal{H}_A^{\aux}$, $V_B: \mathcal{H}_B \to \widetilde{\mathcal{H}}_B \otimes \mathcal{H}_B^{\aux}$ and a unit vector $\ket{\aux} \in \mathcal{H}^{\aux}_A\otimes\mathcal{H}^{\aux}_B$ such that the following hold true: 
\begin{align*}
     & V_A \otimes V_B \ket{\psi} \approx_{\epsilon} \ket{\widetilde{\psi}}\ket{\aux}, \text{ and }\\
 & (V_A \otimes V_B)(E_{x,a} \otimes F_{y,b})\ket{\psi} \approx_{\epsilon} (\widetilde{E}_{x,a} \otimes I_{\mathcal{H}_A^{\aux}}) \otimes (\widetilde{F}_{y,b} \otimes I_{\mathcal{H}_B^{\aux}})\ket{\widetilde{\psi}}\ket{\aux}.
\end{align*}
    We first claim and show the following: 
    \begin{align}
        \abs{\ (\widetilde{\pi}_A(\alpha) \otimes I_{\mathcal{H}_{B}})\ket{\widetilde{\psi}}\ket{\aux} - (V_A \pi_A(\alpha) V_A^* \otimes I_{\widetilde{\mathcal{H}}_B} \otimes I_{\mathcal{H}_B^{\aux}})\ket{\widetilde{\psi}}\ket{\aux}\ } \leq (n+5)l(\alpha)\epsilon
    \end{align}

    We proceed by induction on $l(\alpha)$. For the case $l(\alpha) = 1$, take $\alpha = e_{x,a}$. Then, for an arbitrary choice of $y \in X$ one has 
    \begin{align*}
       (\widetilde{E}_{x,a} \otimes I_{\mathcal{H}_{B}})\ket{\widetilde{\psi}}\ket{\aux} & = \sum_{b} (\widetilde{E}_{x,a} \otimes \widetilde{F}_{y,b})\ket{\widetilde{\psi}}\ket{\aux} \\
       & \approx_{n\epsilon} \sum_{b} (V_A \otimes V_B) (E_{x,a} \otimes F_{y,b}) \ket{\psi} \\
       & =  (V_A \otimes V_B) (E_{x,a} \otimes I_{\mathcal{H}_B}) (V_A \otimes V_B)^*(V_A \otimes V_B) \ket{\psi} \\
       & \approx_{\epsilon} (V_A \otimes V_B) (E_{x,a} \otimes I_{\mathcal{H}_B}) (V_A \otimes V_B)^* \ket{\widetilde{\psi}}\ket{\aux} \\
       & = (V_A E_{x,a} V_A^* \otimes V_B V_B^*)\ket{\widetilde{\psi}}\ket{\aux}
    \end{align*}

    Lastly, since $(V_A \otimes V_B)(V_A \otimes V_B)^*$ is a projection, one has 
    \begin{align*}
        (V_A \otimes V_B)(V_A \otimes V_B)^*\ket{\widetilde{\psi}}\ket{\aux} & \approx_{\epsilon} (V_A \otimes V_B)(V_A \otimes V_B)^* (V_A \otimes V_B) \ket{\psi} \\
        & = (V_A \otimes V_B) \ket{\psi} \approx_\epsilon \ket{\widetilde{\psi}}\ket{\aux},
    \end{align*}
    so that
    \begin{align*}
        (V_A E_{x,a} V_A^* \otimes V_B V_B^*)\ket{\widetilde{\psi}}\ket{\aux} & = (V_A E_{x,a} V_A^* \otimes I) (I \otimes V_BV_B^*)\ket{\widetilde{\psi}}\ket{\aux} \\
        & \approx_{2\epsilon} (V_A E_{x,a} V_A^* \otimes I) (I \otimes V_BV_B^*)(V_A \otimes V_B)(V_A \otimes V_B)^*\ket{\widetilde{\psi}}\ket{\aux}\\
        & = (V_A E_{x,a} V_A^* \otimes I) (V_AV_A^* \otimes V_BV_B^*)\ket{\widetilde{\psi}}\ket{\aux} \\
        & \approx_{2\epsilon} (V_A E_{x,a} V_A^* \otimes I_{\widetilde{\mathcal{H}}_B} \otimes I_{\mathcal{H}_B ^{\aux}})\ket{\widetilde{\psi}}\ket{\aux}
    \end{align*}
    Since the choice of $x,a$ was arbitrary, this establishes the base step of induction. 

    Now, let us assume that this holds true for monomials of length $\leq k$, and consider a monomial $\alpha$ of length $k+1$. Let $\alpha = \alpha' e_{x,a}$ for some $e_{x,a}$ and monomial $\alpha'$ of length $k$. We then have the following: 
    \begin{align*}
        (\widetilde{\pi}_A(\alpha) \otimes I_{\widetilde{\mathcal{H}}_{B}})\ket{\widetilde{\psi}}\ket{\aux} & = (\widetilde{\pi}_A(\alpha') \widetilde{E}_{x,a} \otimes I_{\widetilde{\mathcal{H}}_{B}})\ket{\widetilde{\psi}}\ket{\aux} \\
        & = (\widetilde{\pi}_A(\alpha') \otimes \widetilde{F}_{y,b})\ket{\widetilde{\psi}}\ket{\aux} \\
        & = (I_{\widetilde{\mathcal{H}}_A} \otimes \widetilde{F}_{y,b}) (\widetilde{\pi}_A(\alpha') \otimes I_{\widetilde{\mathcal{H}}_{B}})\ket{\widetilde{\psi}}\ket{\aux} \\
        &\approx_{{(n+5)k\epsilon}} (I_{\widetilde{\mathcal{H}}_A} \otimes \widetilde{F}_{y,b} \otimes I_{\mathcal{H}_B^{\aux}}) (V_A{\pi}_A(\alpha')V_A^* \otimes I_{\widetilde{\mathcal{H}}_{B}} \otimes I_{\mathcal{H}_B^{\aux}})\ket{\widetilde{\psi}}\ket{\aux}\\
        &= (V_A{\pi}_A(\alpha')V_A^* \otimes \widetilde{F}_{y,b} \otimes I_{\mathcal{H}_B^{\aux}})\ket{\widetilde{\psi}}\ket{\aux} \\
        & = (V_A{\pi}_A(\alpha')V_A^* \widetilde{E}_{x,a} \otimes I_{\widetilde{\mathcal{H}}_{B}} \otimes I_{\mathcal{H}_B^{\aux}})\ket{\widetilde{\psi}}\ket{\aux} \\
        & \approx_{(n+5)\epsilon}(V_A{\pi}_A(\alpha')V_A^* V_AE_{x,a}V_A^* \otimes I_{\widetilde{\mathcal{H}}_{B}} \otimes I_{\mathcal{H}_B^{\aux}})\ket{\widetilde{\psi}}\ket{\aux} \\
        & = (V_A \pi_A(\alpha) V_A^* \otimes I_{\widetilde{\mathcal{H}}_B} \otimes I_{\mathcal{H}_B^{\aux}})\ket{\widetilde{\psi}}\ket{\aux},
    \end{align*}
which establishes our claim. Similarly, we can show that the following holds true: 
\begin{align}
        \abs{\ (I_{\mathcal{H}_{B}} \otimes \widetilde{\pi}_B(\beta))\ket{\widetilde{\psi}}\ket{\aux} - (V_B \pi_B(\alpha) V_B^* \otimes I_{\widetilde{\mathcal{H}}_B} \otimes I_{\mathcal{H}_B^{\aux}})\ket{\widetilde{\psi}}\ket{\aux}\ } \leq (n+5)l(\beta)\epsilon
    \end{align}

Hence, we get
\begin{align*}
    \widetilde{\rho}(\alpha \otimes \beta) & = \braket{\widetilde{\psi} \mid \widetilde{\pi}_A(\alpha) \otimes \widetilde{\pi}_B(\beta)\mid \widetilde{\psi}} \\
    & = \braket{\widetilde{\psi},\aux \mid \widetilde{\pi}_A(\alpha) \otimes I_{\mathcal{H}_A^{\aux}} \otimes \widetilde{\pi}_B(\beta) \otimes I_{\mathcal{H}_B^{\aux}}\mid \widetilde{\psi},\aux}\\
    & \approx_{(n+5)l(\alpha)\epsilon} \braket{\widetilde{\psi},\aux \mid V_A \pi_A(\alpha) V_A^* \otimes \widetilde{\pi}_B(\beta) \otimes I_{\mathcal{H}_B^{\aux}} \mid \widetilde{\psi},\aux}\\
    & \approx_{(n+5)l(\beta)\epsilon} \braket{\widetilde{\psi},\aux \mid V_A \pi_A(\alpha) V_A^* \otimes V_B \pi_B(\beta) V_B^*\mid \widetilde{\psi},\aux}\\
    & = \braket{\widetilde{\psi},\aux \mid (V_A \otimes V_B) (\pi_A(\alpha) \otimes \pi_B(\beta)) (V_A \otimes V_B)^*\mid \widetilde{\psi},\aux}\\
    & \approx_{2\epsilon} \braket{\psi\mid (V_A \otimes V_B)^* (V_A \otimes V_B) (\pi_A(\alpha) \otimes \pi_B(\beta)) (V_A \otimes V_B)^*(V_A \otimes V_B) \mid \psi } \\
    & = \braket{\psi\mid \pi_A(\alpha) \otimes \pi_B(\beta)\mid \psi } \\
    & = \rho(\alpha \otimes \beta),
\end{align*} 
which finishes the proof.    
\end{proof}

\begin{theorem}\label{thm:main-theorem-3}
    Let $\widetilde{p}$ be a synchronous correlation that is an extreme point of $C_q(X, X, A, A)$. If $\widetilde{p}$ robustly self-tests a PME correlation $(\mathbb{C}^d, \mathbb{C}^d, \{\widetilde{E}_{x,a}\}, \{\widetilde{F}_{y,b}\}, \ket{\psi})$, then there is a unique (finite-dimensional state) $\widetilde{\rho}: \in S(\mathcal{A}_{XA}^{POVM} \otimes \mathcal{A}_{XA}^{POVM})$ such that $\widetilde{\rho}(e_{x,a} \otimes e_{y,b}) = \widetilde{p}(a,b \mid x,y)$ for all $x,y \in X$ and $a,b \in A$. 
\end{theorem}
\begin{proof}
    let $\widetilde{\rho}$ be the finite dimensional state on $\mathcal{A}_{XA}^{POVM} \otimes \mathcal{A}_{XA}^{POVM}$ induced by the model $\widetilde{S}$, and $\rho$ be any state in $S(\mathcal{A}_{XA}^{POVM} \otimes \mathcal{A}_{XA}^{POVM})$. Then, there exist quantum correlations $p^{[k]} \in C_q(X,X,A,A)$ with models $S^{[k]}$, such that $\lim_{k \to \infty} \norm{\widetilde{p} - p^{[k]}}_1  = 0$. Moreover, if $\rho^{[k]}$ are the induced finite dimensional states on $\mathcal{A}_{XA}^{POVM} \otimes \mathcal{A}_{XA}^{POVM}$ then $\rho^{[k]} \to \rho$ in the weak-$*$ topology. Since $\widetilde{p}$ is a robust self-test for $\widetilde{S}$, it follows from \cref{lem:local-dilation-implies-equal-state} that for any words $\alpha, \beta$ in $\{e_{x,a}\}$ and $\epsilon > 0$, there is a $N \in \mathbb{N}$ such that for all $k > N$, one has $\abs{\rho^{[k]}(\alpha \otimes \beta) - \widetilde{\rho}(\alpha \otimes \beta)} \leq \epsilon$. Since $\alpha, \beta$ and $\epsilon$ were arbitrary it follows that $\widetilde{\rho} = \rho$, which finishes the proof. 
\end{proof}

We note that the proofs of \cref{lem:local-eps-dil-implies-closeness-of-state} and \cref{thm:main-theorem-3} can be generalised to the case where $\widetilde{S}$ is a centrally supported strategy (as defined in \cite{paddock_operator-algebraic_2023}). In particular, since each correlation has a full-rank strategy, the above result holds true for all correlations thus proving \cref{prop:robust-self-test-implies-abstract-self-test}. 

\section{Game Self-tests}

We shall first prove the following result:

\begin{prop}\label{prop:state-correlation-correspondence-qc-synch}
    Let $\pi: \mathcal{A}_{XA}^{POVM} \otimes_{\max} \mathcal{A}_{XA}^{POVM} \to \mathcal{B(H})$ be a representation and $\ket{\psi}\in \mathcal{H}$ be a unit vector, and let $\rho$ be the state $a \mapsto \braket{\psi\mid \pi(a) \mid \psi}$. If the correlation $(\rho(e_{x,a} \otimes e_{y,b}))$ is synchronous, then the following hold true: 
    \begin{enumerate}[label = \roman*.]
        \item $\pi(e_{x,a} \otimes 1)\ket{\psi} = \pi(1 \otimes f_{x,a})\ket{\psi}$ for all $x \in X$, $a \in A$. 
        \item Let $\op: \mathcal{A}_{XA}^{POVM} \to (\mathcal{A}_{XA}^{POVM})^{\op}$ denote the isomorphism from $\mathcal{A}_{XA}^{POVM}$ in to its opposite algebra. Then, for any $a,b \in \mathcal{A}_{XA}^{POVM}$, one has 
        \begin{equation*}
            \rho(a\otimes b) = \rho(a\op(b)\otimes 1) = \rho(\op(b)a \otimes 1) = \rho(1 \otimes b\op(a)) = \rho(1 \otimes \op(a)b) = \rho(\op(a) \otimes \op(b)).
        \end{equation*}
        In particular, the state $a \mapsto \rho(a \otimes 1)$ on $\mathcal{A}_{XA}^{POVM}$ is a trace. 
    \end{enumerate}
\end{prop}
\begin{proof}
    Consider $\norm{\pi(e_{x,a})\ket{\psi} - \pi(f_{x,a})\ket{\psi}}$, 
    \begin{align*}
        \norm{ \pi(e_{x,a} \otimes 1)\ket{\psi} - \pi(1 \otimes f_{x,a})\ket{\psi} } & = \braket{\psi \mid \pi(e_{x,a}^2 \otimes 1)\mid \psi} + \braket{\psi \mid \pi(1 \otimes f_{x,a}^2)\mid \psi} - 2 \braket{\psi \mid \pi(e_{x,a} \otimes f_{x,a}) \mid \psi}  \\
        & \leq \braket{\psi \mid \pi(e_{x,a}\otimes1 )\mid \psi} + \braket{\psi \mid \pi(1\otimes f_{x,a})\mid \psi} - 2 \braket{\psi \mid \pi(e_{x,a} \otimes f_{x,a}) \mid \psi} \\
        & = \braket{\psi \mid \pi(e_{x,a}\otimes\textstyle \sum_{b \in A} f_{x,b})\mid \psi} + \braket{\psi \mid (\textstyle \sum_{b' \in A}\pi(e_{x,b'} \otimes f_{x,a})\mid \psi}\\ & \hspace{2in} - 2 \braket{\psi \mid \pi(e_{x,a} \otimes f_{x,a}) \mid \psi} \\
        & = \braket{\psi \mid \pi(e_{x,a} \otimes f_{x,a})\mid \psi} + \braket{\psi \mid \pi(e_{x,a} \otimes f_{x,a})\mid \psi}\\ & \hspace{2in}- 2 \braket{\psi \mid \pi(e_{x,a} \otimes f_{x,a}) \mid \psi} = 0,
    \end{align*}which establishes \textit{i.}. 
    Note that this gives us $\pi((e_{x_1, a_1}\otimes 1) \cdots (e_{x_n, a_n}\otimes 1))\ket{\psi} = \pi((1 \otimes f_{x_n, a_n})\cdots (1 \otimes f_{x_1, a_1})) \ket{\psi}$. Hence, then for any $a \in \mathcal{A}_{XA}^{POVM}$, one has $\pi(a \otimes 1) \ket{\psi} = \pi(1 \otimes\op(a)) \ket{\psi}$. Hence, 
    \begin{align*}
        \rho(a \otimes b) & = \braket{ \psi \mid \pi(a \otimes 1) \pi(1 \otimes b) \mid \psi} \\
        & = \braket{\psi \mid \pi(a \op(b) \otimes 1) \mid \psi } = \rho(a\op(b) \otimes 1)
    \end{align*}
    The rest of the proof of \textit{ii.} can be completed similarly. 
\end{proof}

One thing we should note from the proof \cref{prop:state-correlation-correspondence-qc-synch} is that if $\rho$ is a state on $\mathcal{A}_{XA}^{POVM} \otimes_{\max} \mathcal{A}_{XA}^{POVM}$, then $\rho_A, \rho_B: \mathcal{A}_{XA}^{POVM} \to \mathbb{C}$ are tracial states defined by $\rho_A(a) = \rho(a \otimes 1)$ and $\rho_B(b) = \rho(1\otimes b)$ and one has $\rho_A(a) = \rho_B(\op(a))$. In a similar fashion to \cref{prop:state-correlation-correspondence-qc-synch}, one can also show the following:

\begin{prop}\label{prop:state-trace-corr}
    Let $\rho: \mathcal{A}_{XA}^{POVM} \otimes_{\max} \mathcal{A}_{XA}^{POVM} \to \mathbb{C}$ be a state such that the correlation $(\rho(e_{x,a} \otimes e_{y,b}))$ is synchronous, then the following hold true:
    \begin{enumerate}[label = \roman*.]
        \item $\pi(e_{x,a}^2) = \pi(e_{x,a})$ and $\pi(f_{y,b}^2) = \pi(f_{y,b})$ for all $x, y \in X$ and $a,b \in A$. 
        \item If $\rho(e_{x,a} \otimes e_{y,b}) = 0$, then $\pi(e_{x,a})\pi(e_{y,b}) = \pi(f_{x,a})\pi(f_{y,b}) = 0$. 
    \end{enumerate}
\end{prop}

Hence, we have the following correspondence between states on $\mathcal{A}_{XA}^{POVM} \otimes_{\max} \mathcal{A}_{XA}^{POVM}$ implementing perfect strategies for $\mathscr{G}$ and tracial states on $C^*(\mathscr{G})$, for any synchronous nonlocal game $\mathscr{G}$:

\begin{theorem}\label{thm:state-trace-corr}
    Let $\mathscr{G} = (X, A, V)$ be a synchronous nonlocal game. Then, there is a bijective correspondence between tracial states on $C^*(\mathscr{G})$ and states $\rho$ on $\mathcal{A}_{XA}^{POVM} \otimes_{\max} \mathcal{A}_{XA}^{POVM}$ such that $\rho(e_{x,a} \otimes e_{y,b}) = 0$, whenever $V(a,b \mid x,y) = 0$. Moreover, this correspondence restricts to the set of amenable traces on $C^*(\mathscr{G})$ and states $\rho$ on $\mathcal{A}_{XA}^{POVM} \otimes_{\min} \mathcal{A}_{XA}^{POVM}$; and finite-dimensional traces on $C^*(\mathscr{G})$ and finite-dimensional states on $\mathcal{A}_{XA}^{POVM} \otimes_{\min} \mathcal{A}_{XA}^{POVM}$.
\end{theorem}
\begin{proof}
    Let $\rho$ be a state on $\mathcal{A}_{XA}^{POVM} \otimes_{\max} \mathcal{A}_{XA}^{POVM}$ such that $\rho(e_{x,a} \otimes e_{y,b}) = 0$, whenever $V(a,b \mid x,y) = 0$, and consider the GNS representation $(\pi,\ket{\psi})$, where $\ket{\psi} \in \mathcal{H}$. Then, since $(\rho(e_{x,a} \otimes e_{y,b}))$ is a synchronous correlation. By \cref{prop:state-trace-corr} we have that $\pi_A: C^*(\mathscr{G}) \to \mathcal{B(H})$ given by $e_{x,a} \mapsto \pi(e_{x,a})$ is a $\ast$-homomorphism. 
    
    Furthermore, since $(\pi,\ket{\psi})$ is a cyclic representation, it follows from \cref{prop:state-correlation-correspondence-qc-synch} that $(\pi_A, \ket{\psi})$ is a cyclic representation. In particular, it is (unitarily equivalent to) the GNS-representation of the tracial state $\tau_A: \mathcal{A}_{XA}^{POVM} \to \mathbb{C}$ defined by $\tau_A(a) =  \rho(\pi_A(a) \otimes 1)$. Since $\rho(a \otimes b) = \rho(a\op(b) \otimes 1)$, it follows that $\rho \mapsto \rho_A$ is an injective map from the set of states $\rho$ on $\mathcal{A}_{XA}^{POVM} \otimes_{\max} \mathcal{A}_{XA}^{POVM}$ such that $\rho(e_{x,a} \otimes e_{y,b}) = 0$ whenever $V(a,b \mid x,y) = 0$ to the set of tracial states on $C^*(\mathscr{G})$. 

    On the other hand, if $\tau$ is a tracial state on $C^*(\mathscr{G})$, let $(\pi, \ket{\psi})$ be its GNS representation. Let's denote $\pi(e_{x,a})$ by $E_{x,a}$. It follows from the proof of \cite{paulsen_estimating_2016}*{Theorem 5.1} that there are operators $F_{x,a} \in \mathcal{B(H})$ such that $E_{x,a}$ and $F_{y,b}$ commute for all $x,y \in X$ and $a,b \in A$ and $\sum_{b\in A} F_{y,b} = 1$ for all $y \in X$. Hence, if we define $\pi': \mathcal{A}_{XA}^{POVM} \to \mathcal{B(H})$ given by $\pi'(f_{x,a}) = F_{x,a}$, we can construct a representation $\pi \otimes \pi': \mathcal{A}_{XA}^{POVM} \otimes_{\max} \mathcal{A}_{XA}^{POVM} \to \mathcal{B(H})$, and a state $\rho_{\tau}$ given by $a \to \braket{\psi\mid \pi \otimes \pi' (a) \mid\psi}$. In particular, one has $\rho_{\tau}(e_{x,a} \otimes f_{y,b}) = \braket{\psi\mid E_{x,a}F_{y,b}\mid\psi} = \braket{\psi \mid E_{x,a}E_{y,b} \mid \psi} = \tau(e_{x,a}f_{y,b})$. Hence, $\tau \mapsto \rho_{\tau}$ is a well-defined map from traceial states on $C^*(\mathscr{G})$ to the states $\rho$ on $\mathcal{A}_{XA}^{POVM} \otimes_{\max} \mathcal{A}_{XA}^{POVM}$ such that $\rho(e_{x,a} \otimes e_{y,b}) = 0$ whenever $V(a,b \mid x,y) = 0$. 

     We can similarly show that $\rho_{\tau}(a \otimes b) = \tau(a\op(b))$, from which it is easy to see that these maps are inverses of each other, which establishes the bijective correspondence. 

     Now, if $\tau$ is a trace on $C^*(\mathscr{G})$, then it follows from \cite{brown88c}*{Theorem 6.2.7} that the restriction of $\rho_{\tau}$ to $\mathcal{A}_{XA}^{POVM} \otimes_{\min} \mathcal{A}_{XA}^{POVM}$ is continuous if and only if $\tau$ is amenable. Moreover, all finite-dimensional traces are amenable. This finishes the proof.  
\end{proof}

\subsection{Nonlocal Games that Robustly Self-Test a Strategy}

As an easy consequence of \cref{prop:robust-self-test-implies-abstract-self-test}, we can write the following for any nonlocal game: 

\begin{cor}
    Let $\mathscr{G} = (X, Y, A, B, V)$ be any nonlocal game. Then, if $\mathscr{G}$ is a robust self-test, then there is a unique state $\rho$ on $\mathcal{A}_{XA}^{POVM} \otimes_{\min} \mathcal{A}_{YB}^{POVM}$ such that $\rho(e_{x,a} \otimes e_{y,b}) = 0$, whenever $V(a,b \mid x,y)= 0$. 
\end{cor}

For synchronous games, as a consequence of \cref{thm:main-result}, \cref{thm:main-theorem-1} and \cref{thm:state-trace-corr}, we can write the following: 

\begin{theorem}\label{thm:synch-game-self-test-cond}
    Let $\mathscr{G}$ be a synchronous game. Then, the following hold true: 
    \begin{enumerate}[label = \roman*.]
        \item $\mathscr{G}$ is a self-test if and only if $C^*(\mathscr{G})$ admits a unique finite-dimensional trace. 
        \item $\mathscr{G}$ is a robust self-test if and only if $C^*(\mathscr{G})$ admits a unique amenable trace.
        \item $\mathscr{G}$ is a commuting operator self-test if and only if $C^*(\mathscr{G})$ admits a unique trace. 
    \end{enumerate}
    In particular, if $\mathscr{G}$ is a commuting operator self-test for a finite-dimensional strategy, then it is a robust self-test.
\end{theorem}

As a consequence of the reduction from LCS games to synchronous games, we can also give some sufficient conditions for an LCS game to be a robust self-test. First, we state the following result:

\begin{prop}[\cite{lupini_perfect_2020}*{Proposition 4.10}, \cite{goldberg2021synchronous}*{Theorem 4.1}]
    Let $\mathscr{S}$ be a linear constraint system game modulo $d$ with a corresponding synchronous nonlocal game $\mathscr{G}_\mathscr{S}$ and solution group $\Gamma(\mathscr{S})$. The $C^*$-algebra of the game $\mathscr{G}_\mathscr{S}$, $C^*(\mathscr{G}_S)$ is $\ast$-isomorphic to $C^*(\Gamma(\mathscr{S}))/\braket{J - \omega 1}$, where $C^*(\Gamma(\mathscr{S}))$ is the full group $C^*$-algebra of $\Gamma(\mathscr{S})$ and $\omega$ is a primitive $d$th root of unity.
\end{prop}

In particular, if $\Gamma$ has a unique representation where $\pi(J) \neq \omega1$ then $C^*(\mathscr{G}_S)$ has a unique representation. Moreover, if this representation is finite dimensional, then so is the representation of $C^*(\mathscr{G}_S)$. This establishes that $C^*(\mathscr{G}_S) \cong M_d(\mathbb{C})$ for some $d \in \mathbb{N}$, so that $C^*(\mathscr{G}_S)$ admits a unique tracial state that is also finite dimensional. Using similar arguments and combining them with \cref{thm:synch-game-self-test-cond}, we have the following result:

\begin{theorem}\label{thm:suff-cond-lcs-game}
    Let $\mathscr{S}$ be a linear constraint system game modulo $d$, and let $\omega$ be a primitive $d$th-root of unity. Then, the following hold true: 
    \begin{enumerate}[label =\roman*.]
        \item If $\Gamma(\mathscr{S})$ has a unique finite-dimensional irreducible representation where $J \neq \omega 1$, then $\mathscr{G}_\mathscr{S}$ is a self-test. 
        \item If $\Gamma(\mathscr{S})$ has a unique representation where $J \neq \omega 1$, that is also finite-dimensional, then $\mathscr{G}_\mathscr{S}$ is a robust self-test.
    \end{enumerate}
\end{theorem}

We note that these conditions are the same as the ones obtained for robust self-testing in \cite{coladangelo2019robustselftestinglinearconstraint}. However, our result comes at the loss of the explicit dependence of $\delta$ on $\epsilon$ in the definition of self-testing.

\subsection{Separating Robust Self-tests and Commuting Operator Self-Tests}

In this section, we show the existence of a synchronous nonlocal game such that is a robust self-test for finite dimensional strategies but it is not a self-test for commuting operator strategies. It follows from \cref{thm:synch-game-self-test-cond} that a synchronous game is a self-test is a robust self-test for finite dimensional strategies if and only if $C^*(\mathscr{G})$ admits a unique amenable trace. On the other hand, it is a robust self-test if and only if $C^*(\mathscr{G})$ admits a unique tracial state. Hence, we need to construct a nonlocal game, whose $C^*$-algebra admits a unique amenable trace, but admits multiple tracial states. 

For nonlocal games $\mathscr{G}_1$ and $\mathscr{G}_2$, the $\mathscr{G}_1 \vee \mathscr{G}_2$ game was introduced in \cite{mancinska_counterexamples_2023}, which was used to construct several counter examples in self-testing. The $\mathscr{G}_1 \vee \mathscr{G}_2$ allows the players to win if they can win at least one of $\mathscr{G}_1$ and $\mathscr{G}_2$. We shall show that if $\mathscr{G}_1$ is a synchronous game that is a robust self-test and $\mathscr{G}_2$ is a synchronous game that has a perfect qc-strategy but no perfect qa-strategy, then $\mathscr{G}_1 \vee \mathscr{G}_2$ is a robust self-test, but is not a commuting operator self-test. We define this game formally below (only for the case of synchronous games,  we refer the reader to \cite{mancinska_constant-sized_2024} for further details regarding the general case):

\begin{definition}
    Let $\mathscr{G}_1 = (X_1, A_1, V_1)$ and and $\mathscr{G}_2 = (X_2, A_2, V_2)$ be two synchronous games. The $\mathscr{G}_1 \vee \mathscr{G}_2 = (X, A, V)$ game is a synchronous game with $X = X_1 \times X_2$, $A = A_1 \sqcup A_2$ and $V$ defined as follows: 
    \begin{align}
        \begin{aligned}
            V(a,b \mid (x_1, y_1), (x_2, y_2)) = \begin{cases}
                V(a,b \mid x_1, y_1) & \text{ if } a,b \in A_1 \\
                V(a,b \mid x_2, y_2) & \text{ if } a,b \in A_2 \\
                0 & \text{ otherwise.} 
            \end{cases}
        \end{aligned}
    \end{align}
\end{definition}

In other words, the verifier sends a pair of questions $(x_1, x_2)$, $(y_1, y_2)$ from the games $\mathscr{G}_1, \mathscr{G}_2$ to both Alice and Bob respectively. They win the game if they choose to play the same game and answer correctly for the chosen game. It is easy to see that since $\mathscr{G}_1$ and $\mathscr{G}_2$ are synchronous games, so is $\mathscr{G}_1 \vee \mathscr{G}_2$. We now prove the existence of a synchronous game satisfying the hypothesis of \cref{thm:counter-example}. First, we shall need a few lemmas:

\begin{lemma}\label{lem:trace-on-or-game}
    Let $\mathscr{G}_1, \ \mathscr{G}_2$ be synchronous nonlocal games. Then, if there is a trace $\rho$ on $C^*(\mathscr{G}_1 \vee \mathscr{G}_2)$ such that $\rho(e_{(x_1,x_2),a}) > 0$ for some $a \in A_1$, then there is a tracial state $\rho_1$ on $C^*(\mathscr{G}_1)$. Moreover, if $\rho$ is amenable or finite-dimensional, then so is $\rho_1$.  
\end{lemma}
\begin{proof}
    Consider the set of operators $E_{x,a_1} = e_{(x,x_2),a_1}$. It is evident that each $E_{x,a_1}$ is a projection and that $E_{x,a_1}E_{y,a_1'} = 0$ whenever $V_1(a_1,a_1' \mid x,y) = 0$. We claim that $\sum_{a_1} E_{x,a_1}$ is independent of $x$. Indeed, one has the following: 
    \begin{align*}
        \sum_{a \in A_1} E_{(x,x_2), a} & = \left(\sum_{a_1 \in A_1} E_{(x,x_2), a_1}\right)\left(\sum_{a_1 \in A_1} E_{(x',x_2), a_1} + \sum_{a_2 \in A_2} E_{(x',x_2), a_2}\right) \\
        & = \left(\sum_{a_1 \in A_1} E_{(x,x_2), a_1}\right)\left(\sum_{a_1 \in A_1} E_{(x',x_2), a_1}\right) \\
        & = \left(\sum_{a_1 \in A_1} E_{(x,x_2), a_1} + \sum_{a_2 \in A_2} E_{(x,x_2), a_2}\right)\left(\sum_{a_1 \in A_1} E_{(x',x_2), a_1}\right) \\
        & = \sum_{a_1 \in A_1} E_{(x',x_2), a_1}
    \end{align*}

    Hence, one can construct a $*$-homomorphism, say $\pi$, from $C^*(\mathscr{G}_1)$ to the $C^*$-algebra generated by the operators $\{E_{x,a_1}\}_{x \in X_1, a \in A_1}$. Then, the restriction of the state $\rho$ to the $C^*$-algebra generated by the operators $\{E_{x,a_1}\}_{x \in X_1, a \in A_1}$ composed with $\pi$ gives a tracial positive linear functional $\rho'$ on $C^*(\mathscr{G}_1)$. We now need to show that this functional is not trivial. This is indeed true, since $\rho(e_{(x_1,x_2),a}) > 0$. Hence, we can construct a tracial state $\rho_1$ on $C^*(\mathscr{G}_1)$ by normalising $\rho'$. The rest of other assertions follow similarly. 
\end{proof}

\begin{lemma}\label{lem:algebra-of-or-game}
Let $\mathscr{G}_1$ and $\mathscr{G}_2$ be synchronous games, and consider the game $\mathscr{G}_1 \vee \mathscr{G}_2$. For any $x_2 \in X_2$, define $E_{x_1,a_1} = e_{(x_1,x_2), a_1}$ for $x_1 \in X_1$ and $a_2 \in A_2$. Then, $E_{x_1, a_1}$ is independent of the choice $x_2$.  
\end{lemma}
\begin{proof}
    Since $\mathscr{G}_1$ is synchronous, we can write the following for any $x_2, x_2' \in X_2$:
    \begin{align*}
    e_{(x_1, x_2), a_1} & = e_{(x_1,x_2), a_1} \left(\sum_{a_1'\in A_1} e_{(x_1, x_2'),a_1} + \sum_{a_2 \in A_2} e_{(x_1, x_2'), a_2}\right) \\
    & =  e_{(x_1,x_2), a_1} \left(\sum_{a_1'\in A_1} e_{(x_1, x_2'),a_1}\right) \\
    & = e_{(x_1,x_2), a_1} e_{(x_1, x_2'),a_1} \\
    & = \left(\sum_{a_1' \in A_1} e_{(x_1,x_2), a_1'}\right)e_{(x_1, x_2'),a_1} \\
    & = \left(\sum_{a_1' \in A_1} e_{(x_1,x_2), a_1'} + \sum_{a_2' \in A_2} e_{(x_1,x_2), a_2'}\right)e_{(x_1, x_2'),a_1} = e_{(x_1, x_2'),a_1}
    \end{align*}
\end{proof}

\begin{theorem}
    Let $\mathscr{G}_1$ and $\mathscr{G}_2$ be two synchronous games such that $\mathscr{G}_1$ is a synchronous game that is a robust self-test and $\mathscr{G}_2$ is a synchronous game that has a perfect qc-strategy but no perfect qa-strategy, then $C^*(\mathscr{G}_1 \vee \mathscr{G}_2)$ admits multiple tracial states exactly one of which is amenable. 
\end{theorem}
\begin{proof}
    Let $\rho$ be any amenable trace on $C^*(\mathscr{G}_1 \vee \mathscr{G}_2)$, and let $\widetilde{\rho}$ be the unique amenable on $C^*(\mathscr{G}_1)$. Since $\mathscr{G}_2$ has no perfect qa-strategy, $C^*(\mathscr{G}_2)$ admits no amenable traces, by \cref{lem:trace-on-or-game}, one has $\rho(e_{(x_1,x_2),a_2}) = 0$ for all $x_1 \in X_1$, $x_2 \in X_2$ and $a_2 \in A_2$. In particular, $\rho$ is zero on any monomials containing generators of the form $e_{(x_1,x_2),a_2}$ for any $x_1 \in X_1$, $x_2 \in X_2$ and $a_2 \in A_2$. 
    
    Hence, $\rho$ is only non-zero on monomials in the generators $e_{(x_1,x_2),a_1}$ for any $x_1 \in X_1$, $x_2 \in X_2$ and $a_1 \in A_1$, so that the support of $\rho$ is contained in the $C^*$-algebra generated by $\{e_{(x_1,x_2),a_1}\}$ such that $x_1 \in X_1$, $x_2 \in X_2$ and $a_1 \in A_1$. It follows from \cref{lem:algebra-of-or-game} that $e_{(x_1,x_2),a_1} = e_{(x_1, x_2'), a_1}$ for any $x_2, x_2' \in X_2$, so that we can take the support of $\rho$ to be a subalgebra of the $C^*$-algebra generated by $\{e_{(x_1, x_2), a_1}\}$ for a fixed $x_2$. Let us denote this $C^*$-algebra by $\mathcal{A}$. Then, there is a fixed surjective $\ast$-homomorphism $\pi$ from $C^*(\mathscr{G}_1)$ to $\mathcal{A}$, mapping $e_{x_1,a_1} \mapsto e_{(x_1,x_2),a_1}$.  
    
    In particular, $\rho \circ \pi$ is an amenable trace on $C^*(\mathscr{G}_1)$, so that $\rho \circ \pi = \widetilde{\rho}$. If $\rho_1$ and $\rho_2$ are two distinct amenable traces on $C^*(\mathscr{G}_1 \vee \mathscr{G_2})$, then there exist some $a \in \mathcal{A}$ such that $\rho_1(a) \neq \rho_2(a)$. However, since $\pi$ is surjective, this would imply that there exist some $a' \in C^*(\mathscr{G}_1)$ such that $\rho_1(\pi(a')) \neq \rho_1(\pi(a'))$, which is a contradction. Hence, it follows that $C^*(\mathscr{G}_1 \vee \mathscr{G}_2)$ admits a unique amenable trace. 

    It remains to show that $C^*(\mathscr{G}_1 \vee \mathscr{G}_2)$ admits multiple traces. This is easy to see. Let $\rho$ be be a state on $C^*(\mathscr{G})$ and let $(\pi_{\rho}, \ket{\eta_\rho})$ for $\ket{\eta_\rho} \in \mathcal{H}_{\rho}$, denote its GNS representation. Now, define a representation $\pi': C^*(\mathscr{G}_1 \vee \mathscr{G}_2) \to \mathcal{H}_{\rho}$ by $\pi'(e_{(x_1, x_2), a_2}) = \pi_{\rho}(e_{x_2, a_2})$ and $\pi'(a) = 0$ for all other generators of $C^*(\mathscr{G}_1 \vee \mathscr{G}_2)$. It is now easy to see that $a \mapsto \braket{\eta_\rho\mid \pi'(a) \mid \eta_\rho}$ defines a state on $C^*(\mathscr{G}_1 \vee \mathscr{G}_2)$ that is different from the amenable trace we considered earlier.  
\end{proof}

\section{Discussion}

An example of a game that is a self-test, but is not a robust self-test is known from \cite{mancinska_counterexamples_2023}. However, the analogous question for correlations still remains open. Hence, we ask : 

\begin{quest}
    Is there a correlation $p$ that is a self-test, but is not a robust self-test?
\end{quest}

In light of \cref{thm:counter-example}, one can also ask: 
\begin{quest}
    Is there a correlation $p$ that is a robust self-test, but is not a commuting operator self-test?
\end{quest}

As a result of \cref{thm:main-result} and \cref{thm:synch-game-self-test-cond}, these questions can be reduced to questions about states on $C^*$-algebras, thus providing a potential way of solving these questions. 

Another natural question to ask is if the converse of \cref{prop:robust-self-test-implies-abstract-self-test} holds true, i.e., does abstract state self-testing over all states on $\mathcal{A}_{XA}^{POVM} \otimes_{\min} \mathcal{A}_{XA}^{POVM}$ or over all states on $\mathcal{A}_{XA}^{POVM} \otimes_{\max} \mathcal{A}_{XA}^{POVM}$ imply robust self-testing? 

\bibliography{main}

\end{document}